\documentclass[preprint]{elsarticle}

\pdfoutput=1

\usepackage[utf8]{inputenc}
\usepackage{graphicx}
\usepackage{amssymb}
\usepackage{amsmath}
\usepackage{textcomp}
\usepackage{gensymb}
\usepackage{subfig}
\usepackage[ruled,linesnumbered,noend]{algorithm2e}
\usepackage{lineno}

\usepackage{amsthm}
\usepackage{setspace}
\usepackage{cases}
\usepackage{mathtools}

\usepackage{hyperref}

\usepackage[usenames]{color}
\usepackage{url}

\usepackage{cprotect}
\usepackage{hyperref}
\hypersetup{
    colorlinks=true,
    linkcolor=cyan,
    filecolor=magenta,      
    urlcolor=blue,
}

\usepackage{subfig}
\usepackage{enumitem}
\usepackage{epstopdf}
\usepackage[misc,geometry]{ifsym}

\newcommand{\CCW}{\mathrm{CCW}}

\definecolor{mRed}{RGB}{254,74,73}
\definecolor{mBlue}{RGB}{70,87,117}
\definecolor{mygenta}{RGB}{216,0,216}

\def\ray#1{\overset{\mapsto}{#1}}

\newcommand{\pjrC}[1]{#1}
\newcommand{\pjrF}[1]{}

\newcommand{\asbC}[1]{#1}
\newcommand{\asbF}[1]{}

\newcommand{\cidC}[1]{#1}
\newcommand{\cidF}[1]{}

\newcommand{\POLYSET}{\mathcal{S}} 
\newcommand{\ARRANG}{\mathcal{A}} 

\newcommand{\NEIGH}{\mathcal{N}} 

\newcommand{\CONVEX}{\text{CVX}} 

\newcommand{\INTERIOR}[1]{\text{INT}(#1)}

\newtheorem{theorem}{Theorem}[section]
\newtheorem{corollary}{Corollary}[theorem]
\newtheorem{lemma}[theorem]{Lemma}

\newcommand{\MCPP}{{\sc mcpp}\xspace}  

\newcommand{\fullSP}{\asbC{FullSP}\xspace}
\newcommand{\CGSP}{\asbC{CGSP}\xspace}

\definecolor{mRed}{RGB}{254,74,73}
\definecolor{mBlue}{RGB}{70,87,117}
\definecolor{mGreen}{RGB}{86,227,159}
\definecolor{darkgreen}{rgb}{0,0.5,0}
\definecolor{darkorange}{rgb}{1.0,0.3,0}

\newcounter{mod}
\newenvironment{model}{\vspace{-1ex}\stepcounter{mod}%
\bgroup

	\begin{subequations}%
		}%
	{\end{subequations}%
\egroup
}

\begin{document}
\begin{frontmatter}

\title{%
Solving the Minimum Convex Partition of Point Sets with Integer Programming}
\author[4]{Allan Sapucaia}
\ead{allansapucaia@gmail.com}

\author[4]{Pedro J. de Rezende\corref{cor}}
\cortext[cor]{Corresponding author}
\ead{rezende@ic.unicamp.br}

\author[4]{Cid C. de Souza}
\ead{cid@ic.unicamp.br}

\address[4]{Institute of Computing, University of Campinas, Brazil}

\begin{abstract}
\cidC{\pjrC{The partition of a problem into smaller sub-problems satisfying certain properties is often a key ingredient in the design of divide-and-conquer algorithms. For questions related to location, the partition problem can be modeled, in geometric terms, as finding a subdivision of a planar map -- which represents, say, a geographical area -- into regions subject to certain conditions while optimizing some objective function.}}
\cidC{\pjrC{In this paper, we investigate one of these geometric problems known as the  Minimum Convex Partition Problem (\MCPP).}}
A \textit{convex partition} of a point set $P$ in the plane is a subdivision of the convex hull of $P$ whose edges are segments with both endpoints in $P$ and such that all internal faces are empty convex polygons.
The \MCPP is an NP-hard problem where one seeks to find a convex partition with the least number of faces.

We present a novel polygon-based integer programming formulation for the \MCPP, which leads to better dual bounds than the previously known edge-based model. Moreover, we introduce a primal heuristic, a branching rule and a pricing algorithm. The combination of these techniques leads to the ability to solve instances with twice as many points as previously possible while constrained to identical computational resources.
A comprehensive experimental study is presented to show the impact of our design choices.
\end{abstract}

\end{frontmatter}
%

\section{Introduction}
\label{sec:intro}





Partitioning problems constitute a fundamental topic in Computational Geometry. One of the best studied among them is the Triangulation Problem where we are given a point set $P$ in the plane, and the goal is to partition its convex hull into triangles using line segments with endpoints in $P$. There are many variations of this problem that optimize different metrics such as segment length and minimum angle~\cite{bern93}.
\cidC{\pjrC{The widespread applicability of triangulations stems from a core idea of the {\em divide-and-conquer} paradigm: if a problem is too complex to be solved at once, break it into smaller, more tractable, sub-problems. Geometric problems tend to benefit from spacial subdivisions generated from their input where sub-problems can more easily be dealt with by considering the faces of the resulting arrangement.
In two dimensions, for instance, it is desirable that these faces be regions satisfying properties that can be explored to make the algorithms more efficient.
Therefore, an often desired structure for the faces is that they be convex polygons.
The smallest convex polygons being triangles, it is understandable why questions regarding triangulations are so vastly investigated.
For an in-depth discussion on triangulations, including structural properties, algorithms and applications, we refer to the book of de Loera et al.~\cite{deloera-triang}.}}
However, all triangulations of a given point set have the same number of triangles and edges and they may be too much of a refinement among the possible convex subdivisions.
\cidC{\pjrC{Hence, from the perspective of divide-and-conquer algorithms, when it comes to the number of sub-problems to be dealt with, which is strongly related to the overall complexity, a triangulation might be just as good as another.

This is where the relevance of the Minimum Convex Partition Problem (\MCPP) stands out.  The \MCPP is a generalization of the triangulation problem in that one seeks a partition of the convex hull of a point set into the minimum number of convex polygons -- some of which might even be triangles.}}

\subsection{Our Contribution}
\cidC{\pjrC{Besides our earlier work~\cite{barboza2019}, there are only a few attempts to solve the \MCPP exactly, and, to date, no other comprehensive experimentation has been reported in the literature.
By casting the problem as an integer linear program (ILP) and designing several algorithmic strategies, we were able, in~\cite{barboza2019}, to solve instances of up to 50 points.

In the present paper, we explore further the usage of ILP to compute optimal solutions for the \MCPP.}}
Our main contribution is a new integer programming formulation for the \MCPP, whereby we solve to provable optimality instances of up to 105 points, \cidC{\pjrC{while using the same amount of computational resources as in our previous work, thus more than doubling the size of the largest instances with known optimum.}}
To achieve this, we devise a primal heuristic and a branching rule, \cidC{\pjrC{and show their effectiveness through experimentation. Also, since the number of variables in the new proposed model is exponential in the cardinality of the input point set, we resort to the use of column generation, which leads to the development of a {\em branch-and-price} algorithm.}}
\cidC{Since ILP is often applied in Operations Research, but \pjrC{much less frequently in Computational Geometry, this article can be seen as a further contribution towards bridging these two communities~\cite{pjr:CanoDO:2015, StephanF2016, deRezende2016, Fekete2017,Fekete2018, pjr:ZambonEJOR:2019, pjr:TozoniTOMS:2016}}}.

\subsection{Literature Review}
The \MCPP has been studied from different perspectives in the literature, including the development of exact and
approximation algorithms, heuristics and theoretical bounds on the optimal value.

Fevens et al.~\cite{fevens-DAM-2001} proposed a dynamic programming formulation for the problem. Let $h$ be the depth of $P$, defined as the number nested convex hulls that need to be removed from $P$ before it becomes empty. Their algorithm has time complexity of $O(n^{3h+3})$, therefore exponential in $h$, which can be as large as $\Theta(n)$. Spillner et al.~\cite{spillner2008} proposed another exact
algorithm with complexity $O(2^k k^4 n^3 + n \log n)$, where $k$ is
the number of points of $P$ in the interior of $CH(P)$.

A compact ILP formulation for the \MCPP based on the construction of a planar subdivision representing a partition obtained by selecting edges from $E(P)$ was proposed by Barboza et al.~\cite{barboza2019}. The authors   presented empirical results showing that, when fed as input to a state-of-the-art ILP solver, it could solve instances with up to 50 points in general position to provable optimality. They also show how to use the linear relaxation of the formulation to find good  heuristic solutions for instances with up to 105 points. Those instances were made publicly available\cite{bsr2018}.

The best approximation algorithm for the \MCPP, with  factor of $\frac{30}{11}$, was proposed by Spillner et al.~\cite{knauer-swat-2006}.

Bounds on the value of optimal solution as a function of $n$ were also studied. Let $F(n)$ denote the maximum cardinality of the minimum convex partition among all instances of size $n$ in general position. The tightest known bounds for $F(n)$ are $\frac{12 n}{11}-2 \le F(n) \le \frac{10n-18}{7}$, where the lower bound was shown in~\cite{garcia2006} and the upper bound in~\cite{lara2004}.

The 2020 Computational Geometry Challenge (CGSHOP)~\cite{cgshop} motivated the advancement of the state-of-the-art for the \MCPP from both theoretical and heuristic points of view.
When the challenge was announced in September 2019, the complexity of the \MCPP was still open and the only known empirical study was by Barboza et al.~\cite{barboza2019}. The challenge consisted in finding good solutions for 346 instances, with sizes ranging from 10 to 1,000,000 points and with different sets of instances, including sets with a large number of collinear points. Those instances and the best solution found were made publicly available. Details about the competition and the teams' progresses are discussed in~\cite{cgshop}. The top three competitors proposed heuristic solutions based on local search.

In November 2019, Grelier~\cite{grelier2019} announced a proof of NP-hardness for the case when the point set is not 
in general position. Their proof relies heavily on the construction of instances with a large number of points lying on the same straight line.

\subsection{Organization of the text}
This paper is organized as follows. Section~\ref{sec:model} describes a polygon-based formulation for the \MCPP with an exponential number of variables. In Section~\ref{sec:colgen}, we address the issue of the number of variables being exponential by describing a column generation approach to solve its linear relaxation. Section~\ref{sec:bnp} explains how we incorporate column generation into a branch-and-pricing framework in order to solve the problem to integrality, which includes a branching rule and implementation details.
Computational experiments and their corresponding results are discussed in Section~\ref{sec:exp}.

\subsection{Basic Notation}
A {\it polygon $p$} of size $t$ can be defined as a cyclic sequence of $t$ distinct points in the plane $p_0, p_1, \cdots, p_{t-1}$, called the {\it vertices} of $p$. Two consecutive points $p_i$ and $p_{i+1}$ of $p$ define a line segment $\overline{p_ip_{i+1}}$, called an {\it edge of $p$}, with addition taken module $t$. We say that $p_i$ and $p_{i+1}$ are the {\it endpoints} or {\it extremes} of the edge $\overline{p_ip_{i+1}}$. Sometimes, for practicality, we consider polygons as given by their cyclic sequence of edges $(\overline{p_0p_1}, \overline{p_1p_2}, $\dots$, \overline{p_{t-2}p_{t-1}},\overline{p_{t-1}p_0})$.

A polygon is called {\it simple} if the intersection between two distinct edges is empty unless they are consecutive, in which case they only share an endpoint. This essentially means that no two edges have a proper crossing. We refer to the sequence of edges of a simple polygon as its {\em boundary}.

Since the boundary of a simple polygon $p$ constitutes a closed planar curve, by the Jordan Curve Theorem, it divides the plane into an unbounded and a bounded region. The latter is called the {\it interior} of $p$, denoted $\INTERIOR{p}$, while the former is the {\em exterior of $p$}. Two polygons are said to be {\it interior-disjoint} if their interiors do not overlap.

Given a sequence of three points $(k, l, m)$ in the plane, we
say that they are {\em colinear}, {\em positively oriented}, or {\em negatively oriented} depending on the value of the cross product $\overrightarrow{kl} \times \overrightarrow{lm}$ being 0, positive or negative.
To simplify notation we write $\CONVEX(k,l,m)=\;$true (or simply $\CONVEX(k,l,m)$) when $(k,l,m)$ is positively oriented.
A geometrical interpretation is that a sequence $(k,l,m)$ is positively (negatively) oriented when we make a left (right) turn at $l$ as we traverse segment $kl$ followed by $lm$.

We say that a polygon $p$ is given in {\em counterclockwise} (CCW) order if $\INTERIOR{p}$ is always to the left as one traverses the edges of $p$ in the order given. It can be proved that this may be checked by verifying that $\CONVEX(p_{j-1}, p_{j}, p_{j+1})$ is true when $p_{j}$ is the lowest vertex of smallest abscissa.
We assume that all simple polygons are given in CCW order.
A simple polygon $p$ is {\it convex} if $\CONVEX(p_{i-1}, p_{i}, p_{i+1})$ for all $0\leq i < t$.
For convenience, given a simple polygon $p$, we may often employ the term polygon to also refer to the union of the boundary and the interior of $p$.

Given a set $P$ of $n$ points in the plane and a polygon $p$, we say that $p$ is {\em empty} with respect to (w.r.t.)~$P$, if $p$ contains no points of $P$ in its interior. When $P$ is understood from the context, we simply say that $p$ is empty.
Denote by $\POLYSET(P)$ the set of all convex polygons with vertices in $P$ that are empty w.r.t.~$P$, and by $CH(P)$ the convex hull of $P$. In this paper, we assume that the sets of points are in {\it general position}, i.e., no three points lie on the same line. Moreover, henceforth all polygons referred to will be convex, unless stated otherwise.

A set $U \subseteq \POLYSET(P)$ of interior-disjoints polygons is called a {\it convex partition} of $P$ if $CH(P) = \bigcup\limits_{p \in U}p$.

Given a set $P$ of $n$ points, let $L(P)$ denote the set of $\Theta(n^2)$ line segments whose endpoints belong to $P$. The {\it complete (geometric) graph} induced by $P$ is $G(P)=(P,E(P))$, where $E(P)=\{\{i,j\}: \overline{ij} \in L\}$. In this text, we refer to a segment $\overline{ij} \in L$ and the corresponding edge in $\{i,j\} \in E(P)$ interchangeably. Similarly, we denote the {\it complete oriented graph} induced by $P$ as $\overrightarrow{G}(P)=(P,A(P))$, where the arcs in $A(P)$ correspond to the two orientations of the edges in $E(P)$.

For each convex partition $U$ of $P$, there is a unique planar graph $G_U=(P,E^U) \subseteq G$, where $E^U$ denotes the set of edges whose line segments belong to polygons in $U$.

The set of line segments $L(P)$ determines a planar subdivision called the {\it arrangement} of $P$. Assuming that $P$ is in general position and $n \ge 3$, this arrangement contains an unbounded face, corresponding to the exterior of $CH(P)$, while all other faces are bounded.
We denote the set of bounded faces of the arrangement of $P$ by $\ARRANG(P)$.
It can be proved that each face $f$ of the arrangement $\ARRANG(P)$ is a convex polygon and that if $f\cap p\neq \emptyset$ for some $p\in \POLYSET(P)$, then ${f} \subseteq {p}$. In this sense, we say that every face of an arrangement $\ARRANG(P)$ is {\it atomic}.
Moreover, a polygon $p \in \POLYSET(P)$ is said to {\em contain a face} $f \in \ARRANG(P)$, denoted by $f \subset p$, if the interior of $p$ contains the interior of $f$.
Lastly, a line or line segment $\ell$ {\it supports} a face $f \in \ARRANG(P)$ if $\ell$ contains one the edges of $f$.

\section{A Set Partition Model for the MCPP}
\label{sec:model}



In this section, we present a new Integer Linear Programming (ILP) model for the \MCPP\ from a standard set partition point of view. Again, let $P$ be a set of $n$ points in the plane.

In this model, we associate a binary variable $u_p$ with each polygon $p \in \POLYSET(P)$ such that polygon $p$ is used in the partition of $CH(P)$ if and only if $u_p = 1$. {This polygon-based approach is different from the edge-based model presented in~\cite{barboza2019}, which will be discussed in Section~\ref{sec:branch}}.

Recall that $\ARRANG(P)$ denotes the set of faces of the arrangement of segments induced by the edges in $E(P)$. We use $f \subset {p}$ to indicate that polygon $p \in \POLYSET(P)$ contains (or covers) face $f \in \ARRANG(P)$.

We now introduce the following Model~\ref{IPsp}:
\begin{model}
	\label{IPsp}
	\begin{align}
	\min \ \ & \sum_{p \in \POLYSET(P)}u_p \label{IPsp_obj}\\
	\text{s.a.} \ \ & \displaystyle \sum_{p \in \POLYSET(P): f \subseteq p}u_p = 1 & \forall f \in \ARRANG(P)\label{IPsp_eq}\\
	& u_p \in \{0,1\} & \forall p \in \POLYSET(P) \label{IPsp_var}
	\end{align}
\end{model}

This model is very straightforward as it has only one family of constraints, namely \eqref{IPsp_eq}, which ensure that each face of the arrangement is covered by a polygon in the solution exactly once. The objective function~\eqref{IPsp_obj} minimizes the number of polygons forming the partition.

The main issue with this model, which will be addressed in the next section, is that the number of polygons in $\POLYSET(P)$ is typically exponential in $n=|P|$, making it impractical to enumerate all variables for large instances.
Besides, the number of constraints in the model, although polynomial in $n$, is also large since an arrangement of $\Theta(n^2)$ line segments can have $O(n^4)$ faces~\cite{aronov92} and each of them corresponds to a constraint in~\eqref{IPsp_eq}.

However, we show next that only a small fraction, $\Theta(n^2)$, of those faces are necessary to ensure that the area within $CH(P)$ is correctly partitioned.

We start by defining necessary concepts and notations. Let $q$ be an arbitrary point in the plane. We denote by $\CCW(i,q)$ the sequence of $n-1$ points in $P\setminus\{i\}$ sorted angularly w.r.t.~the ray $\ray{iq}$ that starts at $i$ an passes through $q$.


Recall that we use $\CONVEX(k,l,m)$ to denote that a sequence of points $(k, l, m)$ is positively oriented.

Let $\overrightarrow{iq}$ denote the \emph{oriented line} that passes through $i$ and $q$ in this order, which divides the plane into two half-planes, called {\it left} and {\it right}. A point $k$ on the left (right) half-plane of $\overrightarrow{iq}$ satisfies  $\CONVEX(i,q,k)$ ($\CONVEX(q,i,k)$).
The remaining points are on $\overrightarrow{iq}$ itself.
We define $\CCW^+(i,q)$ as the prefix of $\CCW(i,q)$ whose points are non-negatively oriented w.r.t.~the oriented line
$\overrightarrow{iq}$. Similarly, we use $\CCW^-(i,q)$ to denote
$\CCW(i,q) - \CCW^+(i,q)$; i.e., the suffix of $\CCW(i,q)$ whose points are negatively oriented w.r.t.~$\overrightarrow{iq}$. We illustrate these concepts in Figure~\ref{fig:ccw}.

\begin{figure}[!htb]
	\centering
	\includegraphics[width=0.45\textwidth]{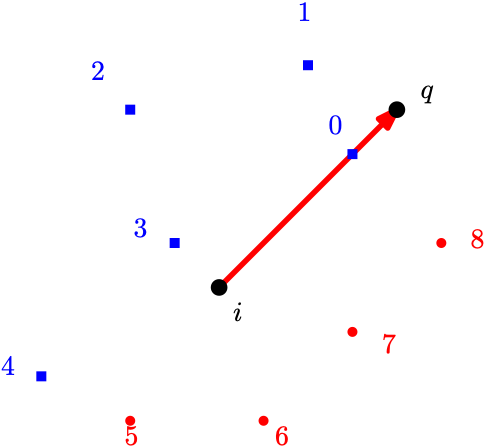}
	\caption{$P=\{i, 0,1,\dots,8\}$, $q\not\in P$,
	$\CCW(i,q)=(0,1,\dots,8)$,
	$\CCW^+(i,q) = (0,1,2,3,4)$ and $\CCW^-(i,q) = (5, 6,7,8)$\label{fig:ccw}}
\end{figure}

Let $i \in P$. A face $f \in \ARRANG(P)$ that has $i$ as one of its vertices is called an {\it $i$-wedge}.
Let $i_x$ be the $x$-coordinate of $i$ and $y_{max}$ be the
maximum $y$ coordinate among all points in $P$ and consider the point $q=(i_x,y_{max}+1)$. Let $\CCW(i,q) = (i_0, i_1, \cdots , i_{n-2})$. For each consecutive pair of points $(i_k, i_{k+1})$ for $ 0 \le k < n-1$, addition being taken mod $n-1$, we have exactly one face of the arrangement $\ARRANG(P)$ that is incident to $i$ and is supported by the edges $\{i,i_k\}$ and $\{i,i_{k+1}\}$. We call this face the {\it $k$-th $i$-wedge}.
We denote by {\it $P$-wedges} the set of all $i$-wedges for $i \in P$.
Figure~\ref{fig:arrang} shows the arrangement of a point set $P$ and highlights the faces that are not $i$-wedges.

\begin{figure}[!htb]
	\centering
	\includegraphics[width=0.5\textwidth]{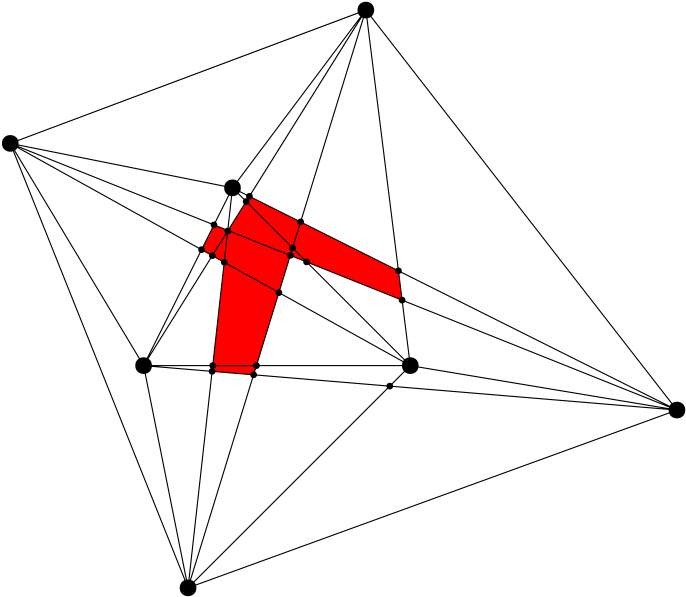}
	\caption{Example of a point set $P$ and its arrangement. Red faces are the only ones that are {\bf not} in $P$-wedges.\label{fig:arrang}}
\end{figure}

Let $U \subset \POLYSET(P)$. We say that a face $f$ of $\ARRANG(P)$ is \emph{exactly-covered} by a polygon $p \in U$, if $p$ is the only polygon in $U$ that covers $f$.
A face that is not exactly-covered by some (single) polygon in $U$ is called \emph{non-exactly-covered}.
A set of faces is exactly-covered if all of its faces are exactly-covered.

In the next lemma, we show that if $u$ is a solution
of Model~\ref{IPsp} with constraint set~\eqref{IPsp_eq} limited to the equations associated to the faces of $\ARRANG(P)$ that are $i$-wedges and $U=\{p \in \POLYSET(P): u_p=1\}$, then $U$ is a partition of $CH(P)$. Our goal is to show that there are no faces in $\ARRANG(P)$ that are non-exactly-covered by polygons in $U$.

\begin{lemma}
Let $u$ be a solution of Model~\ref{IPsp} and $U=\{p \in \POLYSET(P):
y_p=1\}$ be the corresponding set of polygons. Then, if $p$ is a
polygon in $U$, each one of its edges that is not an edge of $CH(P)$
supports exactly one other polygon in $U$. \label{lemma:edge}
\end{lemma}
\begin{proof}
Let $\{i,k\}$ be an edge of $p$ that is not an edge of $CH(P)$, and
$k-1$ be the predecessor of $k$ in $\CCW(i,(i_x, y_{max}+1)))$.
Suppose, w.l.o.g., that the $(k-1)$-st $i$-wedge is exactly-covered by $p$. As $u$ is a solution of Model~\ref{IPsp}, there must be a polygon $h$ in $U$ that covers the $k$-th $i$-wedge. Since we have assumed general position, there is no other edge in $E(P)$ that is supported by the straight line that supports $\{i,k\}$. Also, the interior of $h$ cannot intersect the segment $\{i,k\}$, since, otherwise, both $p$ and $h$ would be covering the $(k-1)$-st $i$-wedge,
violating~\eqref{IPsp_eq}. Therefore, $h$ shares the edge $\{i,k\}$ with $p$. The same holds if we exchange the roles of the $k$-th and the $(k-1)$-st $i$-wedges. Figure~\ref{fig:supportedges} illustrates the polygons in a solution and the $i$-wedges they cover.
\end{proof}
\begin{figure}[!htb]
	\centering
	\includegraphics[width=0.5\textwidth]{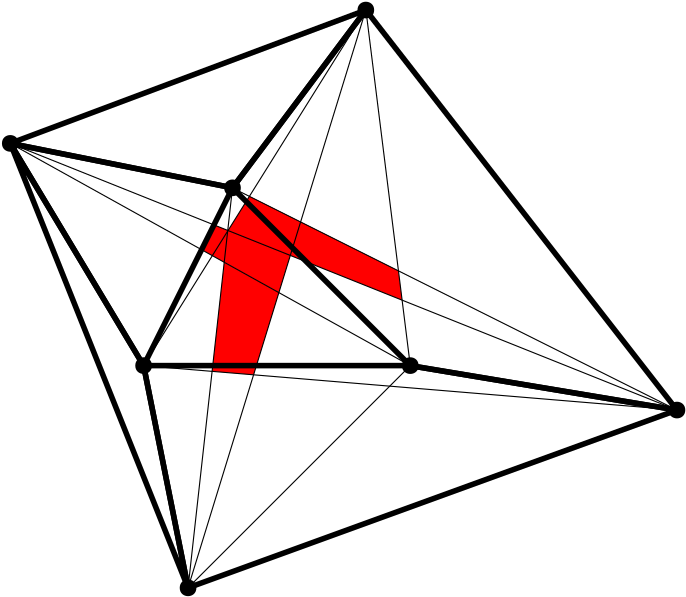}
	\caption{Example of a convex partition superposed over the arrangement of Figure~\ref{fig:arrang}. Edges that belong to the  polygons in the partition are heavier.\label{fig:supportedges}}
\end{figure}

This result helps us prove that constraints that do not correspond to $i$-wedges are, in fact, redundant, and can be removed from the model, as follows.
\begin{theorem}
Let $u$ be a solution of Model~\ref{IPsp} without the Constraints~\eqref{IPsp_eq} that do not correspond to $i$-wedges and $U=\{p \in \POLYSET(P): y_p=1\}$ be the corresponding set of polygons.
Then, $U$ covers each face of the arrangement $\ARRANG(P)$ exactly once.
\end{theorem}
\begin{proof}
First, we show that given a solution $u$ of Model~\ref{IPsp} without the Constraints~\eqref{IPsp_eq} that do not correspond to $P$-wedges, there are no uncovered regions in $CH(P)$.

Notice that the $i$-wedges are exactly-covered. Also, observe that non-exactly-covered regions are formed by unions and intersections of polygons in $U$. Thus, the boundary of each maximal connected non-exactly-covered region is comprised  of of line segments.

Let $R$ be a maximal connected region of $CH(P)$ comprised by the union of uncovered faces of $\ARRANG(P)$. Let $e$ be an edge of $R$ that does not lie on the boundary of $CH(P)$. Such edge always exists since each $i$-wedge is exactly-covered. Thus, $R$ cannot be equal to $CH(P)$. Since $P$ is in general position, $e$ is supported by exactly one edge in $E(P)$, say, $\{i,k\}$. Since $R$ is maximal, there is a face $b$ of the arrangement $\ARRANG(P)$ that is adjacent to $e$ on the other side of $e$ relative to $R$ and is covered at least once. Let $h \in U$ be one of the polygons covering $b$. As $h$ cannot cover any part of $R$, it must include edge $\{i,k\}$ on its boundary. By Lemma~\ref{lemma:edge}, there must be a polygon $g \in U$ that shares $\{i,k\}$ with $h$ on the same side of $e$ as $R$. Thus, $R \cap g$ is covered, which is a contradiction with the fact that $R$ is entirely not covered.

Since $U$ leaves no uncovered regions, we now focus on the proof that there is no region of $CH(P)$ that is covered by two or more
polygons.

Let $R$ be a maximal connected region of $CH(P)$ comprised by the union of faces of $\ARRANG(P)$ that are covered more than once by polygons in $U$. Let $e$ be an edge of $R$ that does not lie on the boundary $CH(P)$. Since $R$ is maximal, one side of $e$ is covered at least twice and the other is covered exactly once, as previously discussed, implying that there is a polygon $g$ in $U$ with an edge $\{i,k\}$ supported by $e$ on the same side as $R$. Let us assume, w.l.o.g., that the $(k-1)$-st $i$-wedge and $R$ are on the same side of the line supported by $\{i,k\}$, otherwise, we change the roles played by $i$ and $k$. Let $b$ be a face of $\ARRANG(P)$ in $R$ adjacent to $e$. Then, $b$ is covered by $g$ and there must be another polygon $h \in U$ that also covers $b$. Since we assumed general position, there is no edge other than $\{i,k\}$ in $E(P)$ that is supported by $e$. However, $h$ cannot have $\{i,k\}$ as one of its edges, since, otherwise, the $(k-1)$-st $i$-wedge would be covered by both $g$ and $h$. Also, since $h$ is empty (w.r.t.~points in $P$), it cannot contain $i$ or $k$ in its interior. This means that polygon $h$ cannot have $e$ on its boundary, and must cover a face $c$ adjacent to $e$ on the opposite side of $g$ relative to $\{i,k\}$. By Lemma~\ref{lemma:edge}, there is a polygon $p$ that shares $\{i,k\}$ with $g$ and covers $c$, implying that $c$ is covered at least twice. We conclude that $R \cup c$ is covered more than once, contradicting our assumption that $R$ is maximal.
\end{proof}

\section{Column Generation Algorithm for SPM}
\label{sec:colgen}

In this section, we address the issue of having an exponential number of variables in Model~\ref{IPsp}. We solved this difficulty by the use of Column Generation.

Instead of enumerating all the variables, we start by solving the
model with only a small subset of them and proceed by generating new ones when necessary. This approach, known as \emph{Column Generation}, guarantees that, if our procedure for generating columns with negative reduced cost is polynomial, the linear relaxation of Model~\ref{IPsp} can be solved in polynomial time\cite{GrotschelLS81}.

To solve the linear relaxation of an LP model using column generation, we iterate between solving the {\it restricted master problem} (RMP) and the {\it pricing problem}. The RMP is the original model restricted to a subset of variables. At each iteration, an RMP is solved to optimality and the pricing problem is used to find variables with negative (in the case of minimization problems) reduced cost that should be added to the RMP. This process stops when no variable has negative reduced cost. Preliminary studies with different set of initial polygons based on the initial solution didn't lead to a noticeable performance difference compared to all triangles. Thus, we start the RMP with the variables corresponding to all triangles.

Next, we show how to solve the pricing problem for Model~\ref{IPsp}.

By associating a vector of variables $\alpha$ with Constraints~\eqref{IPsp_eq}, we obtain the dual of the linear program corresponding to the relaxation of Model~\ref{IPsp}, shown in Model~\ref{IPsp_dual} below.
\begin{model}
\label{IPsp_dual}
	\begin{align}
	\max \ \ & \sum_{f \in \ARRANG}\alpha_f \label{IPsp_dual_obj}\\
	\text{s.a.} \ \ & \displaystyle \sum_{f \in \ARRANG(P): f \subseteq p}\alpha_f \le 1 & \forall p \in \POLYSET\label{IPsp_dual_le}\\
	& \alpha_f \in \mathbb{R} & \forall f \in \ARRANG(P) \label{IPsp_dual_var}
	\end{align}
\end{model}
\vspace{-2ex}

Let $\overline{\alpha}$ be an optimal dual solution of the RMP in a given iteration of the column generation procedure. We can express the reduced cost $\overline{c}_p$ of the variable corresponding to polygon $p \in \POLYSET(P)$ in Model~\ref{IPsp} as:

\begin{equation}
\overline{c}_p = 1-\sum_{f \in \ARRANG(P): f \subseteq p} \overline{\alpha}_f.
\end{equation}
In other words, to compute the reduced cost of a variable, we need to sum up the values of the dual variables associated with face constraints covered by the corresponding polygon.

To solve the pricing problem, we define a recurrence based on the idea of constructing polygons by joining triangles that share an edge.

Let $\Delta(k,l,m)$ denote the reduced cost of a triangle $(k,l,m) \in \POLYSET(P)$, with respect to the dual variables $\overline{\alpha}$, given by $\Delta(k,l,m)=\sum_{f \subseteq (k,l,m)} \overline{\alpha}_f$.

We now consider all polygons whose leftmost vertex $k$ is preceded in CCW order by vertices $l$ and $m$. W.l.o.g., we say that $k$ is the first vertex of such polygons, while $l$ and $m$ are their second-to-last and last vertices, respectively.
Since each of these polygons has an associated variable, let $B(k,l,m)$ denote the minimum reduced cost among these variables. We say that $(k,l,m)$ is the \emph{last triangle} of those polygons. See Figure~\ref{fig:bklm_polygon}.

\begin{figure}[!htb]
	\centering
	\includegraphics[width=0.3\textwidth]{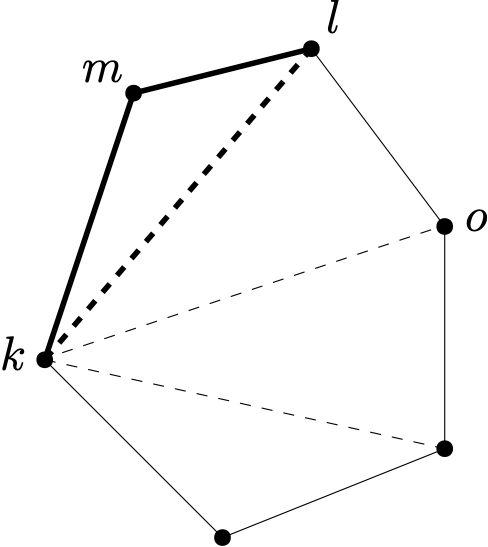}
	\caption{Example of a polygon whose last triangle is $(k,l,m)$.
	Dashed lines show how it can be decomposed into triangles with $k$ as their leftmost vertex, emphasizing the edges of $(k,l,m)$.\label{fig:bklm_polygon}}
\end{figure}

Let us denote $\CCW^-(k, (k_x,y_{max}+1))$ simply by $P_k$, notice that those are the points  to the right of $k$ sorted in CCW order. We use $o \le_{P_k} l$ to indicate that $o$ precedes $l$ in the sequence $P_k$. Recall that $\CONVEX(o,l,m)$ denotes that the sequence $(o,l,m)$ is positively oriented.

Then, $B(k,l,m)$ may 	be computed by following the recurrence formula:
	\begin{subnumcases}{\label{pricing_dp}B(k,l,m)=}
		\infty, \text{ if $m = l$ or $k_x > l_x$ or $m <_{P_k} l$ or $(k,l,m) \not\in \POLYSET$}\label{pricing_dp_invalid}\\
		\min_{\substack{o \in P_k:\\o <_{P_k} l\\ \CONVEX(o,l,m)}}\{0, B(k,o,l)\}+\Delta(k,l,m), \text{ otherwise}\label{pricing_dp_rec}
	\end{subnumcases}

To explain this formula, we define a polygon and a triangle to be \emph{compatible} if their union is a convex polygon that is empty w.r.t.~$P$.
Recurrence~\eqref{pricing_dp} has two cases.
Case~\eqref{pricing_dp_invalid} deals with invalid triplets of
points, while in case~\eqref{pricing_dp_rec} we look for a minimum cost polygon that is compatible with triangle $(k,l,m)$.

Now, we discuss how to solve this recurrence in polynomial time using dynamic programming. Recall that $|P|=n$.

There are $O(n^3)$ dynamic programming states, one for each
$B(k,l,m)$. The total time complexity for processing each state is the $O(n)$ time spent looking for the best compatible polygon, or deciding that $(k,l,m)$ is not empty, plus the time spent calculating $\Delta(k,l,m)$. Naively, computing $\Delta(k,l,m)$ for a single triangle would take, in the worst case, $O(n^4)$ time,
as it would require processing all the faces in the complete arrangement. The final complexity using this approach would
be $O(n^7)$.

However, we can take advantage of the structure of the set of
$P$-wedges to reduce the time to compute $\Delta(k,l,m)$ of a triangle to $O(1)$, with $O(n^2)$ preprocessing each time the dual variables change.

This is accomplished by observing that $\Delta(k,l,m)$ is the sum of a range of consecutive $i$-wedges in each of its vertices $i$, as shown in Figure~\ref{fig:iface_sum}. For each point $i \in P$, we can use a data structure capable of answering range sum queries in $O(1)$ time and which can be built in $O(n)$ time for a given set of dual variable values. Computing $\Delta(k,l,m)$ requires only three such queries. Now, for each of the $O(n^3)$ states, it takes $O(n)$ time to find the compatible polygon with the smallest reduced cost, or decide that $(k,l,m)$ is not empty, and $O(1)$ time to compute $\Delta(k,l,m)$, leading to a time complexity of $O(n^4)$.

\begin{figure}[!htb]
	\centering \includegraphics[width=0.6\textwidth]{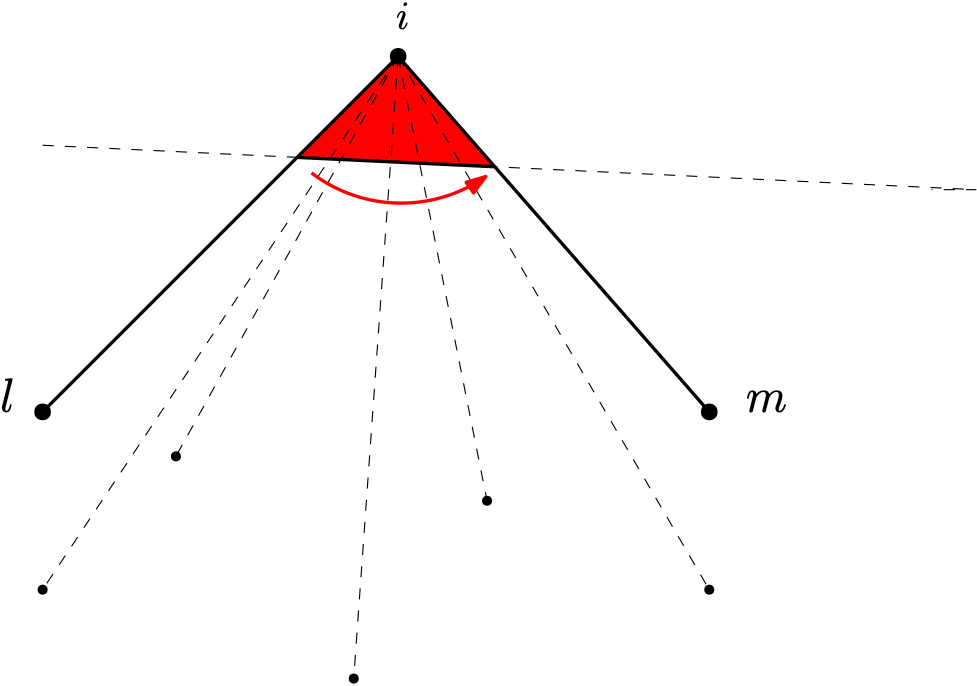}
	\caption{An example of the circular order of $i$-faces around a given vertex $i$. Each pair of consecutive points in a CCW ordering of $P\setminus\{i\}$ defines an $i$-wedge, highlighted in red.\label{fig:iface_sum}}
\end{figure}

To reduce the complexity even further, we adapt the angular sweeping technique presented by Avis and Rappaport~\cite{avis1985} to solve the Largest Empty Convex Polygon Problem and we achieve the same time complexity of $O(n^3)$ as they did for their problem.
Notice that by setting $\Delta(k,l,m)$ to $-1$, finding the polygon with minimum negative cost is the same as finding an empty convex polygon with maximum number of vertices. As done in \cite{avis1985}, we implement an algorithm to check whether a triangle is empty in $O(1)$ time per query and $O(n^3)$ preprocessing time using visibility graphs. However, since we need to solve this problem multiple times, we store all the empty triangles, increasing the space complexity from $O(n^2)$ to $O(n^3)$.

Finally, we present the pricing algorithm, starting with the
description of the angular sweeping technique.

Given $k$ and $l$, let $P_{lk}^+$ be the list $\CCW^+(l,k)$ restricted to the points in $P_k$ and, similarly, $P_{lk}^-$ be the list $\CCW^-(l,k)$ restricted to the points in $P_k$. We compute $B(k,l,m)$ for all possible $m$'s in $O(n)$ total time, using two pointers as follows: $m$ traverses $P_{lk}^-$ and $o$ goes over $P_{lk}^+$. See Figure~\eqref{fig:pricing}.

The procedure to compute the reduced cost $B(k,l,m)$ is presented in Algorithm~\ref{algo:pricing}. To simplify notation, if the triangle $(k,l,m)$ is not empty, we set $\Delta(k,l,m)=\infty$.

\begin{algorithm}
	\DontPrintSemicolon
	\KwIn{Point set $P$, lists of points $P_k$, $P_{lk}^+$ and $P_{lk}^-$}
	\KwOut{Minimum reduced cost $B(k,l,m)$}
	\For{$k \gets 1$ \KwTo $n$}{
		\For{$l \gets 1$ \KwTo $|P_k|$}{
			$o \gets 1$\;
			$bestValueO \gets 0$\;
			\For{$m \gets 1$ \KwTo $|P_{lk}^-|$}{

				\While{$o \le |P_{lk}^+|$ and $\CONVEX(P_{lk}^+[o],P_k[l],P_{lk}^-[m] )$}{
					$bestValueO \gets \min(bestValueO , B(k, P_{lk}^+[o], P_k[l] ))$\;
					$o \gets o+1$\;
				}
				$B(k,l,m) \gets bestValueO + \Delta(k, P_k[l], P_{lk}^-[m])$\;
			}
		}
	}
	\Return{B}\;
	\caption{{\sc Column Pricing Algorithm}}
	\label{algo:pricing}
\end{algorithm}

\begin{figure}[!htb]
    \centering\includegraphics[width=0.4\textwidth]{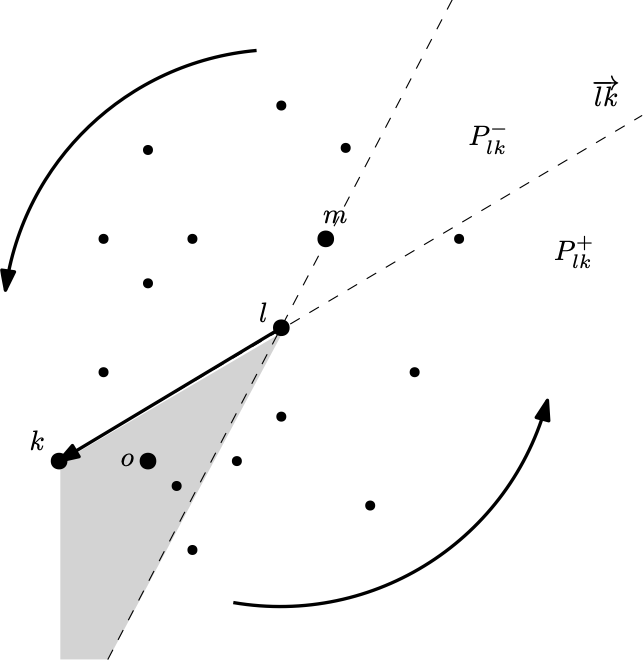}
	\caption{To illustrate the execution of the Column Pricing
	Algorithm \ref{algo:pricing}.
	consider how $P_k$ is
	divided into the lists $P_{lk}^-$ and
	$P_{lk}^+$ by the oriented line $\protect\overrightarrow{lk}$, as well as the movement of
	pointers $l$ and $o$. Points in the gray region are candidates for second-to-last point for polygons whose last triangle is $(k,l,m)$.}
    \label{fig:pricing}
\end{figure}

Algorithm~\ref{algo:pricing} works as follows. The loops defined by Lines 1 and 2 iterate over the leftmost and second-to-last points of each state, respectively. Pointers $m$ and $o$ are initialized to point to the beginning of their respective lists in Lines 3 and 5, and the best compatible polygon is set in Line 4 to be the empty polygon. The loop between Lines 5 and 9 moves pointer $m$, corresponding to the last vertex of the polygons being constructed, one step at the time. The other pointer $o$, corresponding to the third-to-last vertex, is handled by the loop between Lines 6 and 8:
$o$ moves forward as far as possible, while the angle $\measuredangle oml$ is convex and updates the best compatible polygon found so far. Finally, Line 9 computes the cost of $B(k,l,m)$.

To conclude that the total complexity is indeed $O(n^3)$, we need to observe that each one of the variables $k$, $l$ and $m$ always increases and at most $O(n)$ times. Also, each step is done in $O(1)$ time, including computing $\Delta(k,l,m)$, as previously stated.

Notice that reconstructing the polygon given by the dynamic programming table $B(k,l,m)$ takes linear time, since we actually need to find the set of faces that it covers. Doing that for each triplet $k,l,m$ would increase the overall complexity by a factor of $n$. To keep the complexity of the procedure at $O(n^3)$, we only consider the best polygon for each pair $k,l$, limiting the output to $O(n^2)$ polygons.


\section{Branch-and-Price}
\label{sec:bnp}

In the previous section, we discussed how to solve the pricing problem in polynomial time within a Column Generation framework, allowing for the linear relaxation of Model~\ref{IPsp} to be solved in polynomial time. In practice, the standard way to solve ILP models is to use a Branch-and-Bound algorithm based on linear relaxation. The combination of Branch-and-Bound and Column Generation is known as \emph{Branch-and-Price} (BNP)~\cite{barnhart98}.

In this section, we describe some details and the design choices we made to implement a Branch-and-Price algorithm for Model~\ref{IPsp}.

We also discuss how the addition of variables representing edges to the model leads to a specialized branching rule. Moreover, we elucidate how to find an initial viable solution, called an \emph{incumbent} solution, and how to use this solution to generate an initial set of columns for the RMP. A similar approach is used to find viable solutions at each node of the BNP tree. We conclude the section with additional implementation details that aim to improve the performance of the BNP algorithm in practice.

\subsection{Branch on Edges}
\label{sec:branch}
One of the challenges we face when implementing Branch-and-Price algorithms is that branching decisions are handled as additional constraints added to the newly created sub-problems. These constraints have their own dual variables and might change the pricing problem, when the variables being priced are involved.

The idea to introduce edge variables to the Set Partition Model~\ref{IPsp} comes from the Compact Model introduced in~\cite{barboza2019}. Let $S^C$ denote the set of pairs of segment of $E(P)$ that cross and $I(P)$ denote the set of points of $P$ in the interior of $CH(P)$. In the Compact Model~\ref{IPcm}, shown below, we associate an edge variable $x_e$ to each edge $e \in E(P)$.
\vspace{-1ex}
\begin{model}
\label{IPcm}
	\begin{align}
	\min \sum_{\{i,j\} \in E(P)}x_{ij} & & \label{IP_obj}\\
	\text{s.t.}\ \ \ \ \ \ \displaystyle x_{ij} + x_{k\ell} & \le 1 & \forall \{\{i,j\},\{k,\ell\}\} & \in S^c\label{IPcm_col}\\
	x_{ij} & = 1   & \forall \{i,j\} & \in CH(P)\label{IPcm_ch}\\
	\sum_{k \in \CCW^+(i,j)}x_{ik} & \ge 1 & \forall (i,j) \in A(P), i & \in I(P)\label{IPcm_ang}\\
	\sum_{j\in P}x_{ij} & \ge 3 & \forall i & \in I(P) \label{IPcm_deg}\\
	0 \le x_{ij} & \le 1 & \forall \{i,j\} & \in E(P)\label{IPcm_bounds}\\
	x_{ij} & \in \mathbb{Z}   & \forall \{i,j\} & \in E(P)\label{IPcm_integer}
    \end{align}
\end{model}
\vspace{-1ex}

In the Compact Model~\ref{IPcm}, edge crossings are avoided by Constraints~\eqref{IPcm_col}. Constraints~\eqref{IPcm_ch} ensure that the edges belonging to the convex hull of $P$ are part of the solution, while Constraints~\eqref{IPcm_ang} ensure that the angles incident to a internal vertex are all convex. Constraints~\eqref{IPcm_deg} are added to force every internal vertex to have degree at least three. Finally, Constraints~\eqref{IPcm_bounds} and~\eqref{IPcm_integer} guarantee that the variables are binary.

By including edges we can take advantage of the constraints and heuristic presented in~\cite{barboza2019}. The edges are also natural candidates for branching.

A common drawback of Branch-and-Bound algorithms is the possibility of unbalanced branching trees. When one of the branching choices is much more restrictive than the other, the least restrictive node can have almost no impact in the value of the optimal solution of the relaxation~\cite{Foster1976}. This happens, in particular, when branching on a single variable, which is the default for commercial solvers.

For partitioning problems, one of the most well known balanced branching approaches is the Ryan-Foster branching
rule~\cite{Foster1976}. According to this rule, we find a pair of constraints that are covered by distinct sets of variables but share at least one (fractional) variable, and branch on two possibilities: forcing those two constraints to be covered by the same variable or to be covered by two different variables.

In this section, we show that for Model~\ref{IPsp}, when the solution is fractional, there is always a pair of adjacent $i$-wedges that can be used for branching according to the Ryan-Foster rule. This branching can be interpreted geometrically as deciding whether a particular edge of $E(P)$ is part of the solution or not.

Consider edge $ik \in E(P)$ and assume that $k$ is preceded by $k-1$ and succeeded by $k+1$ in the CCW ordering around $i$. This edge has two corresponding arcs $(i,k), (k,i) \in A(P)$, one for each possible orientation. The polygons in $\POLYSET(P)$ can be split into four sets with respect to the arc $(i,k)$: the set of polygons $\POLYSET^+_{ik}(P)$ that are supported by $\{i,k\}$ and cover the $k$-th $i$-wedge; the set of polygons $\POLYSET^-_{ik}(P)$ that are supported by $\{i,k\}$ and cover the $(k-1)$-st $i$-wedge; the polygons $\POLYSET^{over}_{ik}(P)$ that cover both $(k-1)$-st and the $k$-th $i$-wedges; and the set of polygons $\POLYSET^{disj}_{ik}(P)$ that cover neither the $(k-1)$-st nor the $k$-th $i$-wedge. To simplify notation, we omit the point set $P$ from the notation of $\POLYSET$ and its subsets when the context makes it clear. Figure~\ref{fig:polysetsplit} illustrates this notation.
The same reasoning can be used for the reverse arc $(i,k)$, by exchanging the roles played by $i$ and $k$.

\begin{figure}[!htb]
    \centering 
	\includegraphics[width=0.3\textwidth]{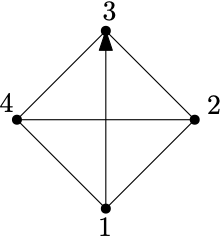}
	\caption{Example of an instance with four points inducing four triangles and one quadrilateral. Arc $(1,3) \in A(P)$ splits those polygons into four sets $\POLYSET^+_{13}=\{134\}$, $\POLYSET^-_{13}=\{123\}$, $\POLYSET^{over}_{13}=\{124, 1234\}$ and $\POLYSET^{disj}_{13}=\{234\}$
	\label{fig:polysetsplit}}
\end{figure}

Recall that the $i$-wedge supported by the edges $i(k-1)$ and $ik$ is the $k$-th $i$-wedge, while the $i$-wedge between edges $ik$ and $i(k+1)$ is the $(k+1)$-st $i$-wedge.

By denoting $\sum_{b \in B}u_b$ as $u(B)$ for a given set $B \subseteq \POLYSET(P)$, we can rewrite equation~\eqref{IPsp_eq} for those $i$-wedges as:
\begin{equation}
\label{eminusrw}
u(\POLYSET^-_{ik})+u(\POLYSET^{over}_{ik})=1
\end{equation}
\begin{equation}
\label{eplusrw}
u(\POLYSET^+_{ik})+u(\POLYSET^{over}_{ik})=1,
\end{equation}
respectively. Notice that equality $u(\POLYSET^+_{ik})=u(\POLYSET^-_{ik})$ follows from~\eqref{eminusrw} and~\eqref{eplusrw}.

With this in mind, we can extend Model~\ref{IPsp} by adding the set of binary variables $x_e$ for each edge $e \in E(P)$. Denoting that an edge $e$ is in the border of a polygon $p$ by $e \in p$, we also add constraints relating polygons and edges.

An edge
${ik}$
is in the solution if and only if there are two polygons supported by it, one on each side. Taking the orientation from $i$ to $k$, we can express this with the following equations:

\begin{equation}
\label{projEdgePolyMinus}
x_{ik} = u(\POLYSET^-_{ik})
\end{equation}
\begin{equation}
\label{projEdgePolyPlus}
x_{ik} = u(\POLYSET^+_{ik}).
\end{equation}

However, since $u(\POLYSET^-_{ik})=u(\POLYSET^+_{ik})$, we can drop one of the constraints or combine both to simplify notation, leading to the following extended Model~\ref{IPext}.
\begin{model}
\label{IPext}
	\begin{align}
	\min \ \ & \sum_{p \in \POLYSET}u_p \tag{\ref{IPsp_obj}}\\
	\text{s.a.} \ \ & \displaystyle \sum_{p \in \POLYSET: f \subseteq p}u_p = 1 & \forall f \in \ARRANG\tag{\ref{IPsp_eq}}\\
	& 2x_e = \sum_{p \in \POLYSET: e \in p} u_p & \forall e \in E(P)\label{IPextProj}\\
	& u_p \in \{0,1\} & \forall p \in \POLYSET \tag{\ref{IPsp_var}}\\
	& x_e \in \{0,1\} & \forall e \in E(p) \label{IPext_varx}
	\end{align}
\end{model}

By branching on an edge variable $x_e$, we are deciding whether the two consecutive $i$-wedges are going to be covered by the same polygon or by different ones. This is analogous to the Ryan-Foster branch rule, as previously discussed.

Lastly, we now show that we can branch only on edge variables.

\begin{lemma}
	Let $(u^*, x^*)$ be an optimal solution of the linear relaxation of Model~\ref{IPext}. Then, there is an edge $e \in E(P)$ such that $x^*_e$ is fractional if and only if there is a polygon $p \in \POLYSET$ such that $u^*_p$ is fractional.
\end{lemma}
\begin{proof}
({\sc counterpositive proof of $\Rightarrow$}): observe that if $u^*_p$ is integral for all $p$, then, if $e$ is any edge in $E(P)$, we either have $u^*(\POLYSET^{over}_{e})=0$, which implies $x_{e}=1$, or $u^*(\POLYSET^{over}_{e})=1$, which ascertains that $x_{e}=0$.

\noindent
({\sc direct proof of $\Leftarrow$}):
Let $p$ be a polygon such that $u^*_p$ is fractional and $e'=\{i',k'\}$ one of its edges. Assume w.l.o.g.~that $p$ covers the $(k'-1)$-st $i'$-wedge, i.e., $p \in \POLYSET^{-}_{e'}$.

\ Case 1: If $u^*(\POLYSET^{over}_{e'}) > 0$ then, as $u(\POLYSET^-_{e'}) > 0$, from Equations~\eqref{eminusrw} and~\eqref{eplusrw}, we have that $u(\POLYSET^+_e)$ and $u(\POLYSET^-_e)$ must both be fractional. Thus, from $0 < u(\POLYSET^+_e)+u(\POLYSET^-_e) < 2$ and~\eqref{IPextProj} we get that $x_e$ is fractional.

\ Case 2: If $u^*(\POLYSET^{over}_{e'}) = 0$, then there is a polygon $h\neq p$ that also covers the $(k'-1)$-st $i'$-wedge such that $0 < u^*_h < 1$.
We can then traverse the edges of $p$ in clockwise order, starting at $e'$, until we find an edge $e=\{i,k\}$ that belongs to $p$ but not to
$h$. Assume w.l.o.g. that $h$ covers both the $(k-1)$-st and the $k$-th $i$-wedge (see Figure~\ref{fig:frac_proof}), otherwise, just
exchange the roles of $p$ and $h$. By construction, $p \in \POLYSET^{-}_{e}$ and $h \in \POLYSET^{over}_{e}$ and the corresponding variables $u_p$ and $u_h$ are both positive.
Therefore, by Equation~\eqref{eminusrw}, $u(\POLYSET^-_e)$ is fractional and, by Equation~\eqref{projEdgePolyMinus}, $x_e$ is also fractional.
\end{proof}

\begin{figure}[!htb]
	\centering\includegraphics[width=0.6\textwidth]{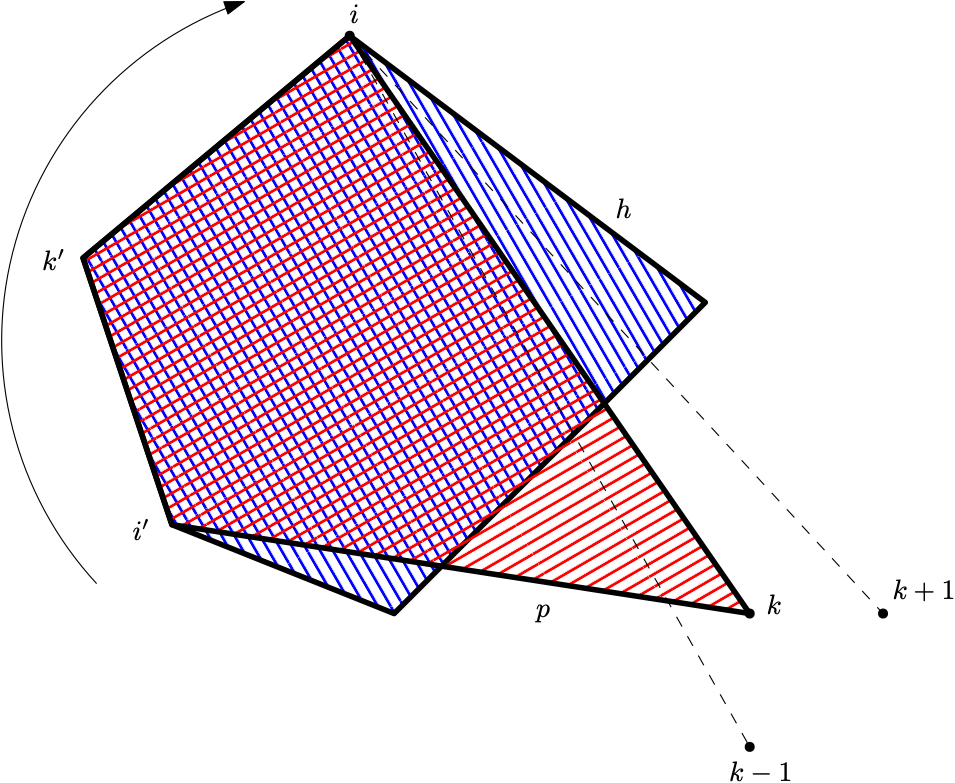}
	\caption{Example of two distinct polygons $p$ and $h$ that share edges. Edge $ik$ belongs only to $p$, but point $i$ belongs to both polygons. Polygon $h$ covers both the $(k-1)$-st and the $k$-th $i$-wedge while $p$ only covers the $k$-th $i$-wedge.\label{fig:frac_proof}}
\end{figure}

Hence, the following Corollary is immediate.
\begin{corollary}
	If $(u^*, x^*)$ is a fractional solution of Model~\ref{IPext}, then there is a pair of adjacent $i$-wedges that induce a Ryan-Foster Branching represented by an edge.
\end{corollary}

Notice that adding variables is not a necessary step to implement the branching rule. We can simply add the corresponding constraints, replacing $x_e$ by zero or one, to the respective child nodes.
However, explicitly adding the variables facilitates the implementation when using a commercial solver as it can take advantage of complex single variable branching rules already in place such as strong branching.

The addition of Constraint~\eqref{IPextProj} affects the pricing algorithm and it needs to be handled explicitly. On the other hand, the actual branching decisions, which add the constraints $x_e=0$ or $x_e=1$ to the child nodes, do not involve polygon variables explicitly. This allows for a slightly modified pricing algorithm that does not change as the branch-and-price progresses.

Let $\beta_e$ be the set of dual variables associated with Constraints~\eqref{IPextProj}. The new reduced costs for the polygon variables in Model~\ref{IPext} are:

\begin{equation}
\overline{c_p} = 1-\sum_{f \in \ARRANG: f \subseteq p}\alpha_f-\sum_{e \in p}\beta_e \mbox{\hspace{2cm}} \forall p \in \POLYSET.
\end{equation}

Now, we need to slightly modify the pricing algorithm to accommodate the addition of these constraints. The reduced cost of a triangle $(k,l,m)$ is set to:
\begin{equation}
\Delta(k,l,m) = \beta_{kl}+\beta_{lm}+\beta_{km}+\sum_{f \subset (k,l,m)} \alpha_f.
\end{equation}

Considering that the $x$ variables represent support edges for the polygons, when building the polygon as the union of a triangle and a compatible polygon, the cost of the shared edge needs to be subtracted twice.

Notice that by setting an edge variable $x_e$ to one, we implicitly forbid any edge that crosses it to be part of the solution. Hence, when making this branching decision, we also explicitly set all crossing edges to zero in the corresponding sub-problems. Besides, to prevent polygons with forbidden edges to be generated, we set the dual cost associated with the constraints corresponding to those edges to an arbitrarily large number.

{Instead of using the default branching rules, we implemented a simple one based on the geometry of the problem. To that end, we say that a variable $x$ in a given solution is \textit{more fractional} the closer the value of $\mathit{frac}(x) = |0.5-x|$ is to zero.
Among all edge variables $x_e$ such that $\mathit{frac}(x_e)$ is within $0.1$ of the most fractional edge variable, we branch on the one whose corresponding edge has the highest number of crossings with all other edges in $E(P)$. This branching rule turned out to be very fast when compared to the default Strong Branching~\cite{achterberg2007} implemented in SCIP.}

\subsection{Primal Heuristic}
\label{sec:heur}
At the end of each BNP node, after pricing is finished, we use the edge variables to construct a greedy triangulation and solve an instance of the MCPP restricted to a small subset of edges, as described in~\cite{barboza2019}.

To build the greedy triangulation $\Delta$, the edges of $E(P)$ are sorted by descending value of their corresponding variables in the current LP solution of the RMP and inserted, if possible, in $\Delta$, following this ordering. An edge is not inserted when it crosses one of the edges previously added to $\Delta$. The result is an inclusion maximal set of non-crossing edges, which characterizes a triangulation.

A viable integer solution is then obtained by solving the MCPP restricted to the edges of this greedy triangulation and its flip edges. This heuristic runs very fast, taking less than 2 seconds for instances of up to 100 points.

It is possible that the solution found by an edge based heuristic contains polygons that have not yet been included in the RMP. This is addressed by running a secondary pricing algorithm whose purpose is solely to add the variables corresponding to the polygons present in heuristic solutions to the RMP.

\subsection{Initial Primal Solution}
To find an initial primal solution solution, we use the same heuristic as in the previous section. However, since no LP solution is available in the beginning, we replace the greedy triangulation with the Delaunay Triangulation.

\subsection{Lower Bound and Early Stopping}
At iteration $t$, the optimal solution of the current RMP $z^t$ is not guaranteed to be a lower bound for Model~\ref{IPext} unless no negative reduced cost variables were found by the pricing algorithm.

When an upper bound on the sum of the variables being priced $\kappa \ge \sum_{p \in \POLYSET}u_p$ is known, a lower bound for the Complete Model~\ref{IPext} can be computed before pricing is finished. If $\overline{c^t}$ denotes the most negative reduced cost obtained during the $t$-th iteration of the pricing algorithm, a lower bound for Model~\ref{IPext} is given by $z^t + \kappa\overline{c^t}$. Since the objective function of Model~\ref{IPext} is $\sum_{p \in \POLYSET}u_p$, the value of any viable integer solution can be expressed as $\kappa$. However, this lower bound is poor, possibly even negative, during the first iterations. According to~\cite{lubbecke05}, a tighter bound $\underline{z^t}$ for the particular case of unitary cost objective functions, which does not depend on the quality of the incumbent solution, is given by:
\begin{equation}
\label{unitdualbound}
\underline{z^t} = \frac{z^t}{1-\overline{c^t_p}}.
\end{equation}

Since all the coefficients in Model~\ref{IPext} are integers, knowing a lower bound for $z^t$ allows for an early halt of the column generation procedure. If $\lceil z^t \rceil = \lceil {\underline z^t} \rceil$, the integer lower bound at the current node cannot be improved by solving the RMP to optimality, and we can proceed directly to branching.

\subsection{Stabilization}
\label{sec:stab}
The lower bound given by~\eqref{unitdualbound} can oscillate between iterations and its convergence to the optimal value of the relaxation might be slow. Improving this convergence can significantly reduce total solving time of the ILP. This can be accomplished by the use of Stabilization techniques.

We can employ the dual bound given by~\eqref{unitdualbound} to assess the quality of a dual solution, where a higher lower bound indicates a better solution. One way to stabilize the algorithm is to minimize drastic changes in the dual solution, keeping it close to a known good solution, also called the \textit{stabilization center}. An in-depth discussion of the topic and different techniques to address the issue can be found in~\cite{pessoa13}.

The dual solution can be stabilized by applying the smoothing technique proposed by Wentges~\cite{wentges1997}. Instead of using the current dual solution $\alpha^t$ in the pricing subroutine, the following convex combination can be considered
\begin{equation}
\alpha^t_{STAB} = \alpha^t+\lambda(\alpha_{BEST}-\alpha^t),
\end{equation}
\noindent
where $0 \le \lambda < 1$ is the smoothing factor and $\alpha_{BEST}$ is the stabilization center corresponding to the dual solution with best lower bound found so far.
{We remark that tuning the parameter for this stabilization was very hard to accomplish, since instances would perform significantly better or worse as the parameter changed, averaging out little change.}

Another way to reduce oscillations is to solve the linear relaxations using barrier methods instead of the commonly used Simplex algorithms. By changing the LP algorithm, we can take advantage of the fact that barrier methods find solutions that lie in the center of the optimal face, as opposed to the extreme points encountered by Simplex. In highly degenerate problems, extreme points can oscillate significantly between iterations with the addition of new columns and/or rows, while subsequent central points are uniquely defined and close together. Another advantage of this approach is that no parameter tuning is necessary\cite{munari13}.

On the other hand, when replacing Simplex with barrier methods, we lose the capability of fast re-optimization after branching, one of the most important features of Dual Simplex algorithms explored when implementing branch and cut procedures\cite{wolsey}. Also, due to the nature of the central solutions found by barrier methods, it is very unlikely that they are integral, increasing the need of a good primal heuristic.

In pure branch-and-cut algorithms, when branching is done or violated cutting planes are found, the addition of new constraints makes the current primal solution infeasible while maintaining its dual feasibility. In this case, the use of Dual Simplex as the LP algorithm to optimize the linear relaxations in the child nodes is recommended because, starting from the current dual feasible basis, it usually requires far fewer iterations to reach a new optimal solution than it would be necessary for a Simplex algorithm started from scratch.

Pricing works similarly, however, the addition of columns makes the current dual solution infeasible while maintaining primal feasibility. So, when combining both row and column generation as in a Branch-and-Price algorithm, if a single Simplex algorithm is used to compute relaxations, some re-optimization is inevitable, reducing the negative impact of switching to barrier methods.

In Section~\ref{sec:exp} we show that, for this particular model, the more stable pricing procedure obtained, by replacing Simplex with a Barrier Method, out-weights the re-optimization cost.

\subsection{Degree Constraints}
\label{sec:degcons}
The addition of edge variables allows for the inclusion of the following degree constraints~\eqref{IPcm_deg} introduced in the Compact Model~\ref{IPcm} for point sets in general position:

\begin{equation}
\sum_{j \in P\setminus\{i\}}x_{ij} \ge 3, \forall i \in I(P)
\label{eq:deg-const}
\end{equation}

Notice that these constraints do not involve polygon variables and, therefore, do not require modifications to the pricing problem.

In practice, we verified that the addition of degree constraints improved the quality of the lower bound provided by the Set Partition Model, and the total solving time. However, when using column generation, the  simple  addition of those constraints resulted in a larger number of calls to the column generation procedure, considerably increasing solving times. 

To minimize this negative side effect, we separate the degree constraints after column generation rather than adding them all at once, despite the fact that there are only $O(n)$ of them. The embedding of a cutting plane procedure in the Branch-and-Price algorithm leads to a Branch-Cut-and-Price algorithm. Also, in our implementation, in an attempt to limit the number of pricing rounds, we require that an inequality be violated by at least $0.1$ units to be inserted in the current linear relaxation. As shown in~\cite{barboza2019}, the degree constraints can increase the lower bound by 0.5 even in very simple instances that, after rounding, may be just enough to assert a known primal solution as optimal, halting the optimization sooner.

\section{Experimental Results}
\label{sec:exp}

In this section, we show the positive impact of some design choices and compare Model~\ref{IPsp} with Model~\ref{IPcm}, presented in~\cite{barboza2019}.

All experiments were run on an Intel Xeon Silver 4114 at 2.2Ghz, and 32GB of RAM running Ubuntu 16.04.
Models and algorithms were implemented in C++ v.11 and compiled with gcc~5.5. Geometric algorithms and data structures were implemented using CGAL~5.1\cite{cgal}, using Gmpq for exact number representation.
The compact model introduced in~\cite{barboza2019} was implemented using CPLEX~12.10, while the set-partition Model~\ref{IPext} used SCIP~7.0\cite{scip} with CPLEX as LP solver. A time limit of \asbC{3 hours} was set for the ILP solver for each instance.

When running times are presented, we consider both the time to generate and to solve the model. Most of the data are presented in a standard boxplot, grouped by size. All data used to generate the figures is available at \cite{bsr2018}.

To compare the algorithms, we employ the instances from~\cite{barboza2019}, available at~\cite{bsr2018}. Those instances had been generated by independently sampling $x$ and $y$ coordinates from a uniform distribution, ensuring general position.

\asbC{In this study, we only use the instances of 65 to 105 points, with 30 instances per size, since smaller instances were too easy, while larger ones were too hard, given our limit of \cidC{3 hours} of (exclusive) solver time.
As we were pushing our models to their limit, we reached, for size 105, a large enough instance size for which multiple failures began to appear.
For instances of 105 points, when more than one instance could not be solved by a given configuration, the solving times for this size are omitted.
Despite the fact that a few of the instances could not be solved to optimality, we include them as clear outliers in some of the forthcoming figures.
Notice that the most basic configuration of Model \ref{IPsp} solved all instances of 65 points in at most 242 seconds, see Figure \ref{fig:timeT2full}, while the compact Model \ref{IPcm} failed to solve any of them in 2400 seconds to provable optimality.}

The main reason for the difference in performance is the quality of the lower bound provided by the relaxations of both models. See Figure~\ref{fig:difCompactFull}.
\begin{figure}[!htb]
\vspace{-3ex}
    \centering\includegraphics[width=\textwidth]{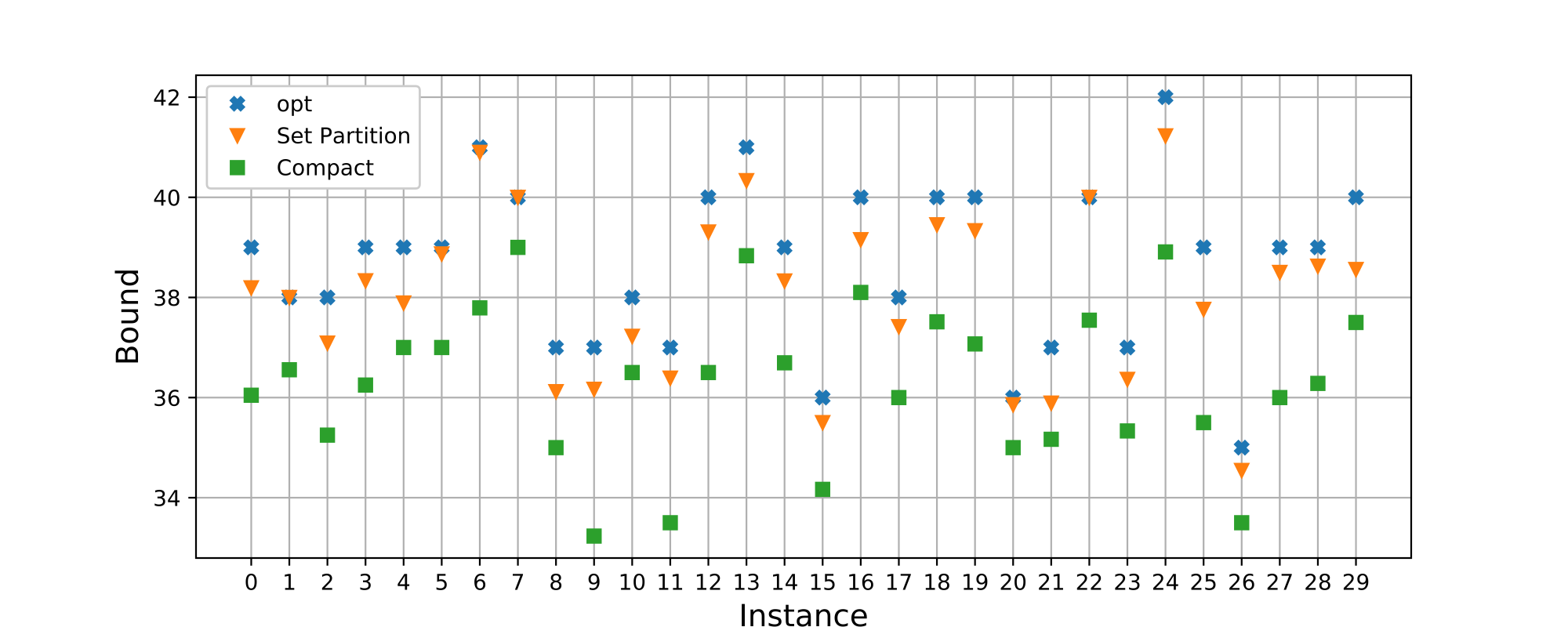} 
	\caption{Lower bound of the relaxations given by the Compact Model and the Full Set Partition Model for instances of size 55 compared with the optimal solution.\label{fig:difCompactFull}}
\end{figure}

\asbC{In the first three experiments, we show the impact of improvements to the Model~\ref{IPsp} with all polygon variables, which we call \textit{\fullSP}, discussed in Section~\ref{sec:bnp}: branching on edges, heuristic, and degree constraints. Our next experiment focuses on Model~\ref{IPext} using column generation, which we call \textit{\CGSP}, comparing the different types of stabilization methods and their impact. We conclude the experiments by comparing the memory consumption of \fullSP\ and \CGSP. Each experiment includes the features introduced in the previous ones.}

\paragraph{Branch on Edge Variables}
We now discuss the impact of introducing edge variables and the branching rule described in Section~\ref{sec:branch} to the basic \fullSP\ model, leading to Model \ref{IPext}.
Figures~\ref{fig:timeT2full} and ~\ref{fig:timeT2edges} show the solving times, respectively, with and without the inclusion of edge variables, while Figures~\ref{fig:nodesT2full} and~\ref{fig:nodesT2edges} depict the number of nodes explored.
Although there is an increase in the number of nodes explored, our implementation is faster than the one based on Strong Branch~\cite{achterberg2007}, the default branching rule for SCIP, leading to better total running times.

\begin{figure}[hbt]
	\centering
	\subfloat[\fullSP\label{fig:timeT2full}]{{\includegraphics[width=.45\textwidth]{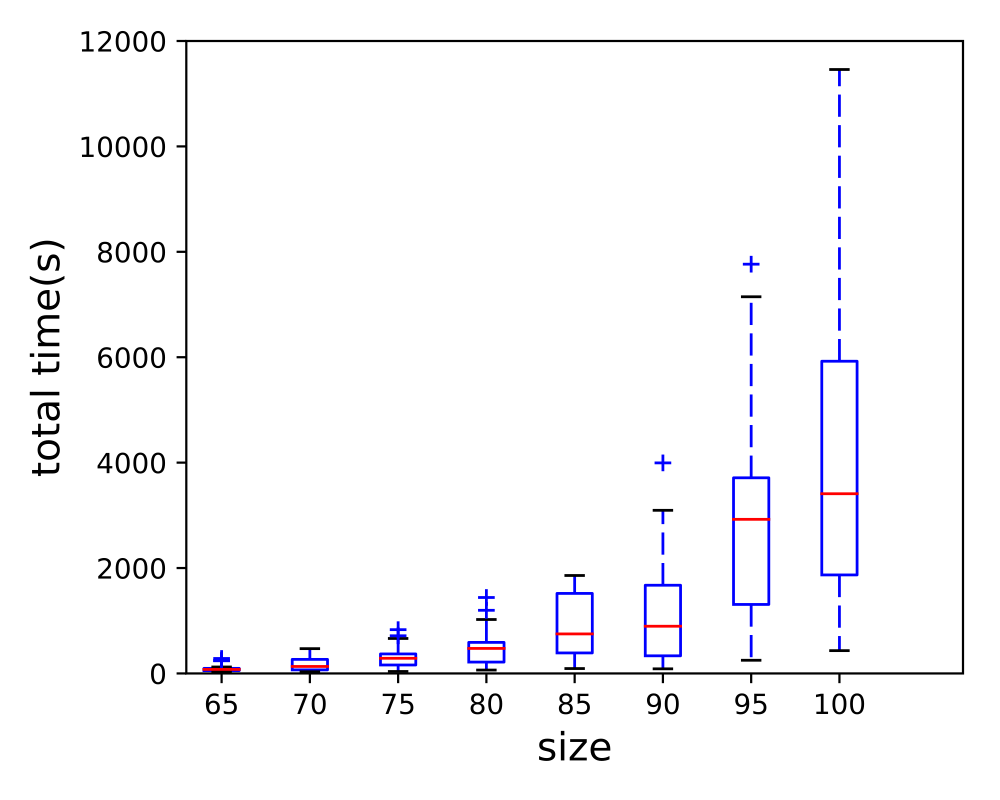}
	}}
	\qquad
	\subfloat[\fullSP+EdgeBranching\label{fig:timeT2edges}]{{\includegraphics[width=.45\textwidth]{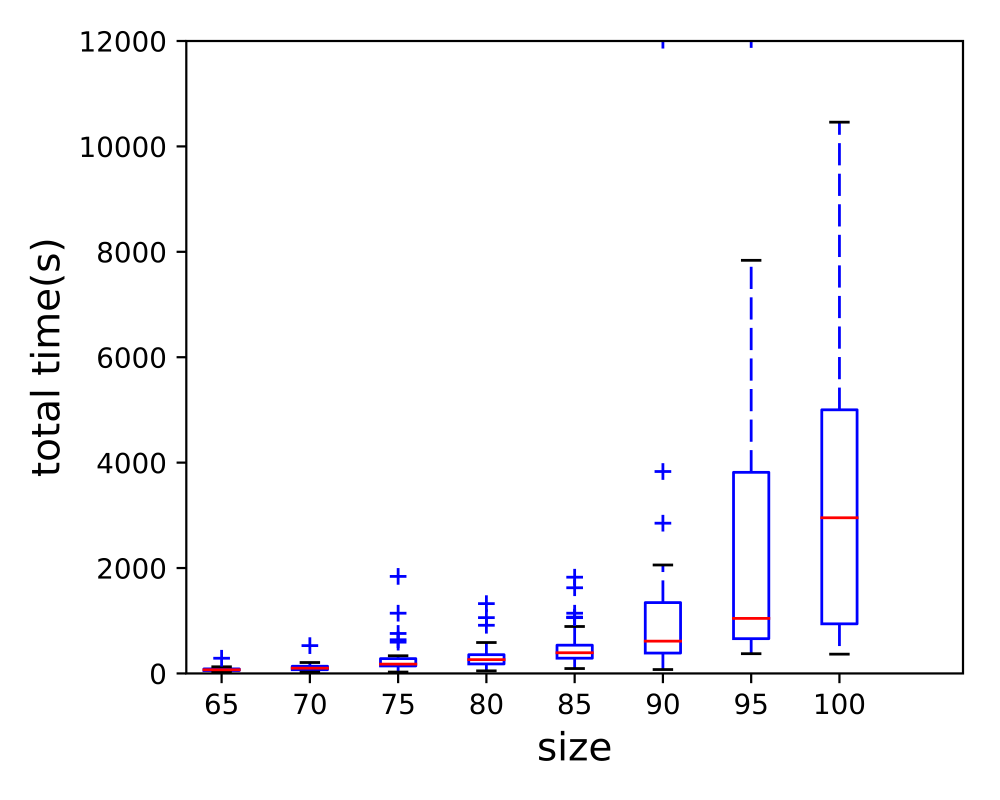}}}%
	
	\subfloat[\fullSP+EdgeBranching+Heur\label{fig:timeT2heur}]{{\includegraphics[width=.45\textwidth]{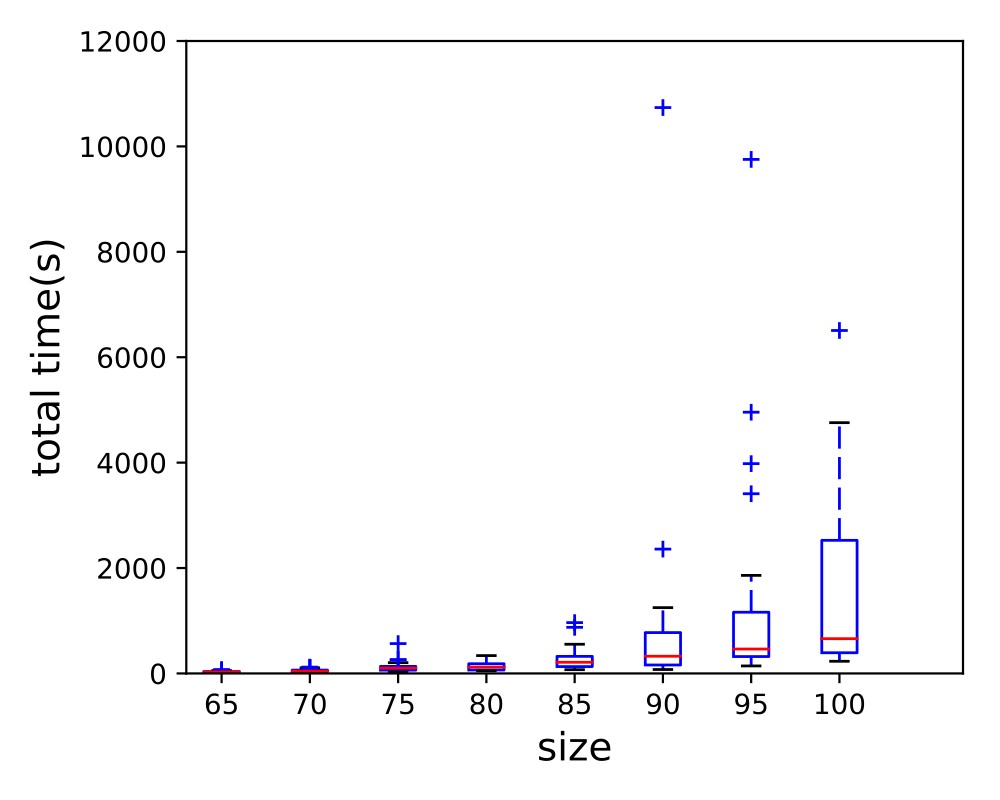}
	}}
	\qquad
	\subfloat[\fullSP+EdgeBranching+Heur+Deg\label{fig:timeT2deg}]{{\includegraphics[width=.45\textwidth]{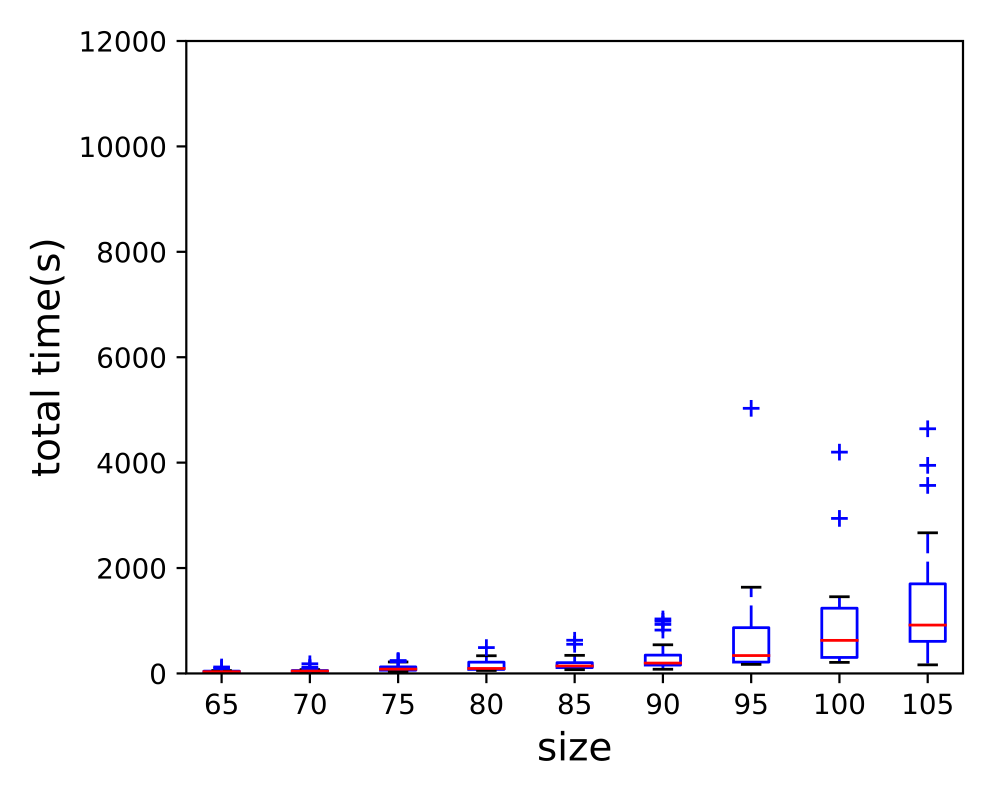}}}%
	\vspace{1ex}
	
	\caption{Total solving time for the \fullSP with different configurations.
	\label{fig:timeFullAll}}%
\end{figure}

\begin{figure}[!hbt]
	\centering
	\subfloat[\fullSP \label{fig:nodesT2full}]
	{{\includegraphics[width=.45\textwidth]{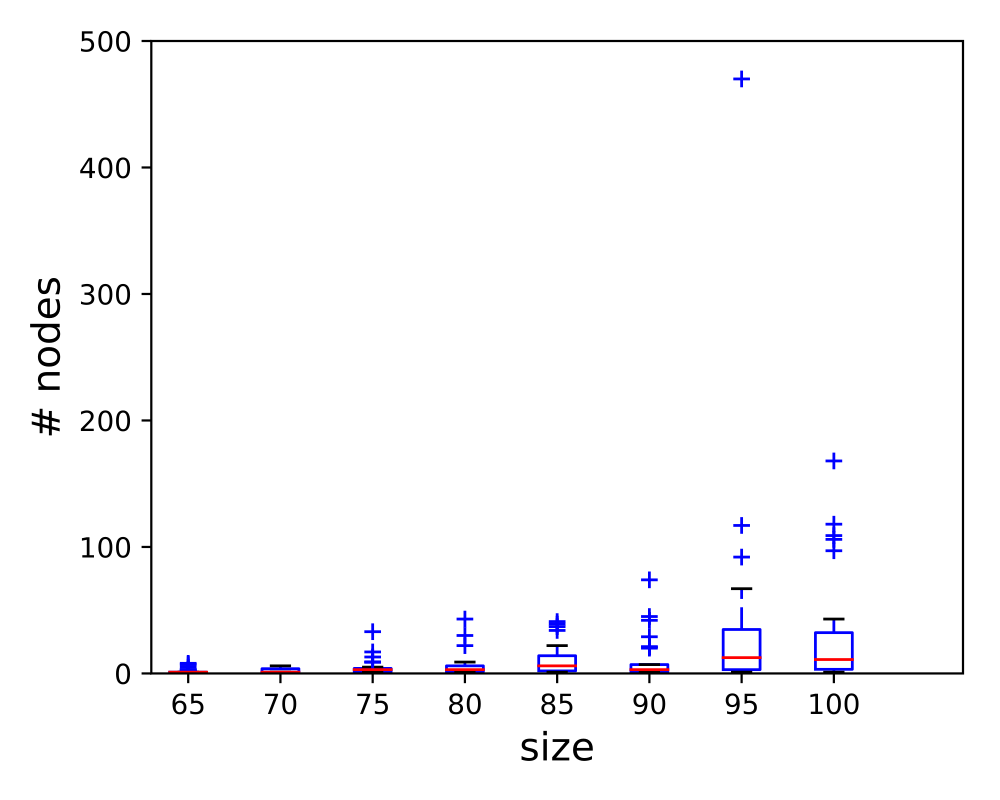}}}
	\hfill
	\subfloat[\fullSP+EdgeBranching\label{fig:nodesT2edges}]
	{{\includegraphics[width=.45\textwidth]{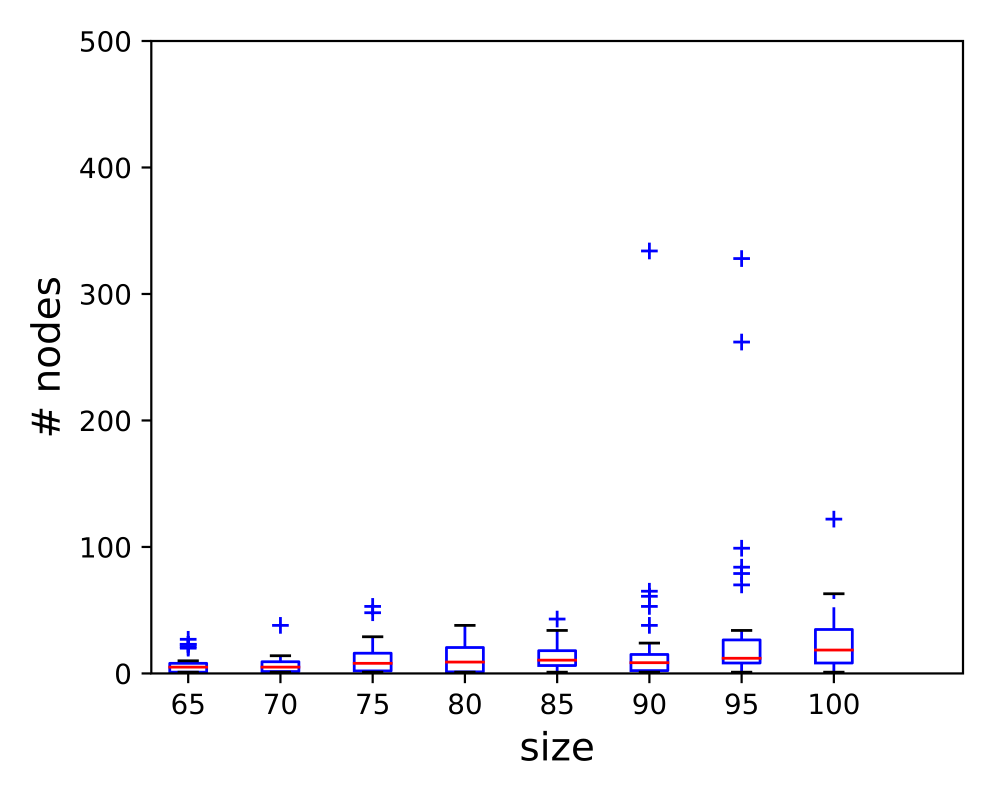}}}%
	\hfill
	\subfloat[\fullSP+EdgeBranching+Heur\label{fig:nodesT2heur}]
	{{\includegraphics[width=.45\textwidth]{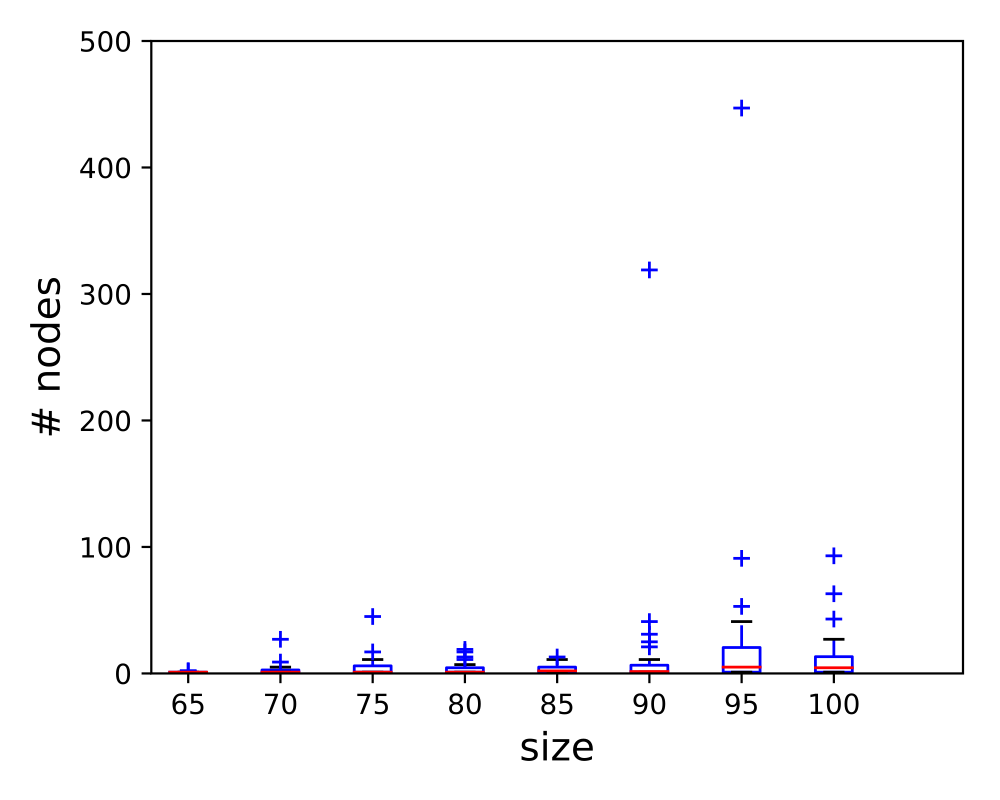}}}
	\hfill
	\subfloat[\fullSP+EdgeBranching+Heur+Deg\label{fig:nodesT2deg}]
	{{\includegraphics[width=.45\textwidth]{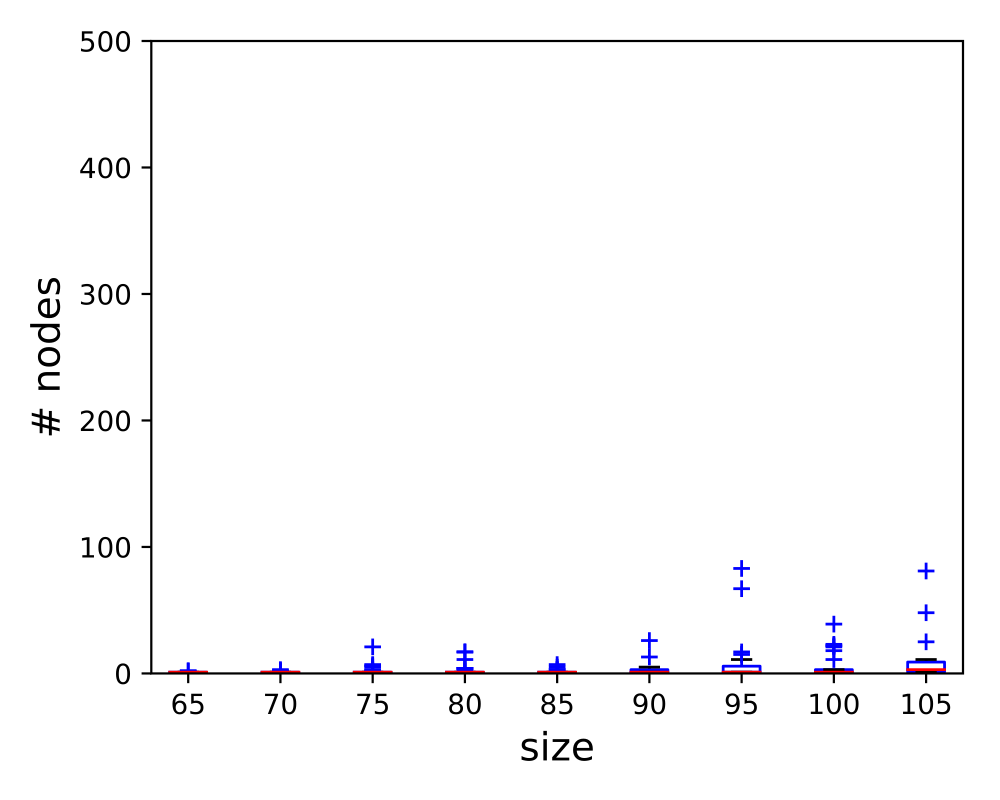}}}
	
	\caption{Number of nodes of the search tree explored during for different configurations of \fullSP. \label{fig:nodesFullAll}}%
\end{figure}

\paragraph{Primal Heuristic}
Next, we investigate the effect of introducing a custom primal heuristic. The positive impact of the inclusion of the primal heuristic described in Section~\ref{sec:heur} is evidenced by comparing Figures ~\ref{fig:timeT2edges} and~\ref{fig:timeT2heur}. The lower bound provided by the model is very strong. Thus, having good heuristic solutions soon considerably reduces the number of nodes explored. See Figures~\ref{fig:nodesT2edges} and~\ref{fig:nodesT2heur}.

\paragraph{Degree Constraints}
Finally, we discuss the impact of the introduction of constraints~\eqref{eq:deg-const}. The performance gain from this inclusion is shown in Figures ~\ref{fig:timeT2heur} and ~\ref{fig:timeT2deg}. Degree constraints considerably increase the strength of the model, even improving the lower bound of simple instances, as discussed in~\cite{barboza2019}.This is shown in Figure~\ref{fig:difHeurDeg}, where lower bounds of the \fullSP, with and without Degree Constraints, are compared with the optimal values. \asbC{This stronger formulation leads to fewer nodes being explored during search, as shown in Figure \ref{fig:nodesT2deg}}. \asbC{Instances with 105 points were solved, on average$\;\pm\;$std-dev, in $1356 \pm 1102$ seconds when using degree constraints}.

\begin{figure}[!htb]
	\centering\includegraphics[width=\textwidth]{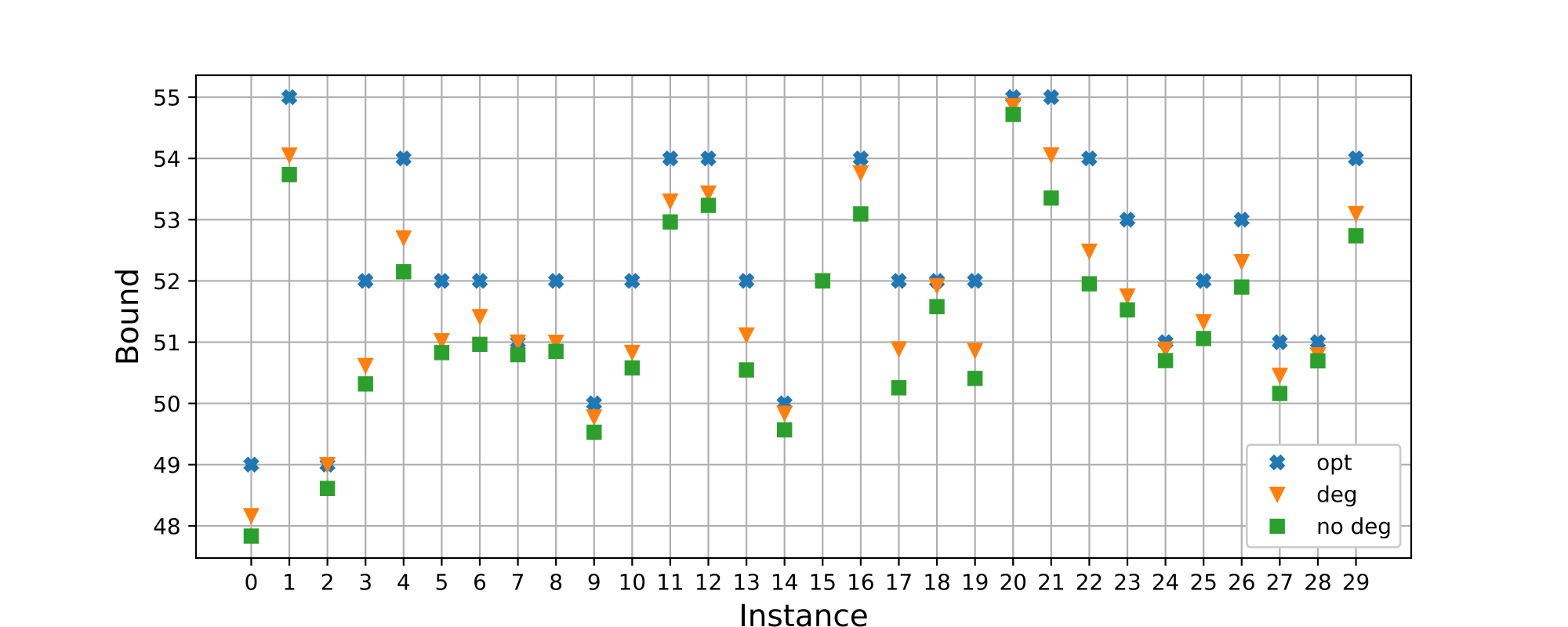}
	\vspace{1ex}
	\caption{Lower bound of the relaxations given by the \fullSP with and without Degree Constraints for instances of size 90 compared with the optimal value. \label{fig:difHeurDeg}}
	
\end{figure}

\paragraph{Column Generation}
All the previously discussed features lead to our best performing \fullSP. However, due to the exponential number of columns, the model eventually becomes too large to fit in memory. In our initial testing, instances with 190 points used more than 32GB of RAM during model creation. For instances with 180 points, the solver started running with less than 2GB left, it is very likely that there would not be enough memory available after branching.

With the memory issue in mind, we try to find the best configuration for the \CGSP.
The RMP is initialized with the set of all triangles and the polygons that belong to the initial heuristic solution.

As mentioned in Section~\ref{sec:degcons}, the simple inclusion of the degree constraints significantly worsen the performance of the model. We overcame this issue by implementing a cutting plane procedure to separate degree constraints. Since the number of constraints is $O(n)$, the separation is made by inspection, only including cuts with a significant violation of at least $0.1$. 

\begin{figure}[!hbt]
	\centering
	\subfloat[No Stab\label{timeT2tri}]{{\includegraphics[width=.45\textwidth]{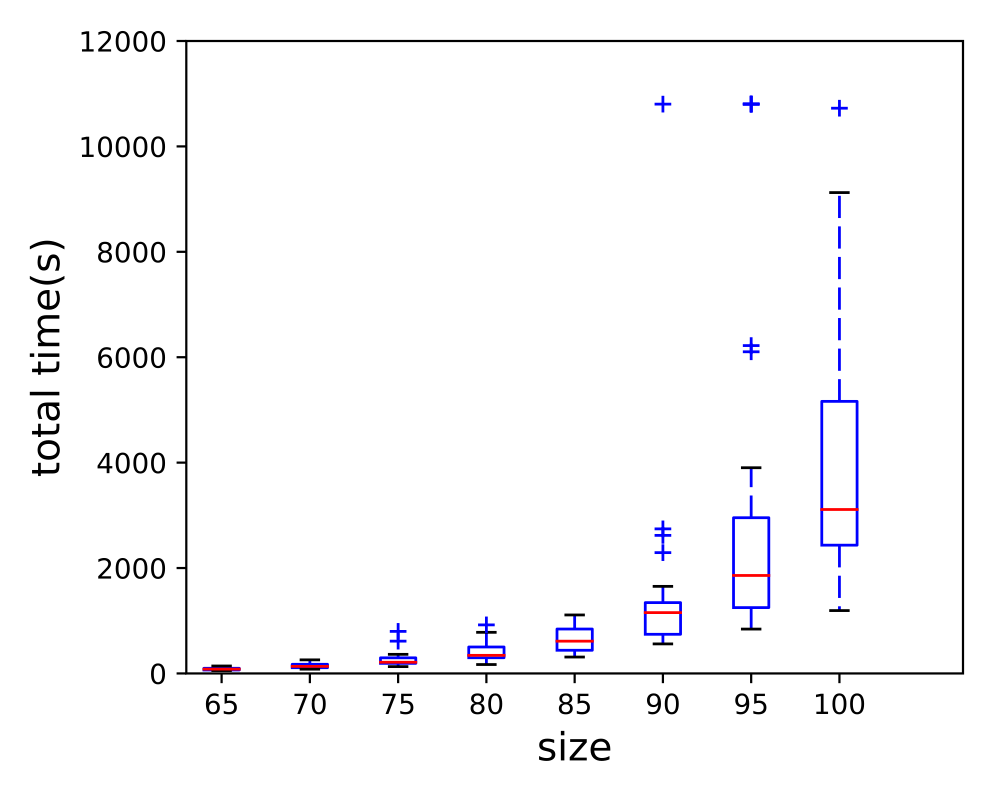}
	}}
	\hfill
	\subfloat[Weingetz\label{timeT2K1}]{{\includegraphics[width=.45\textwidth]{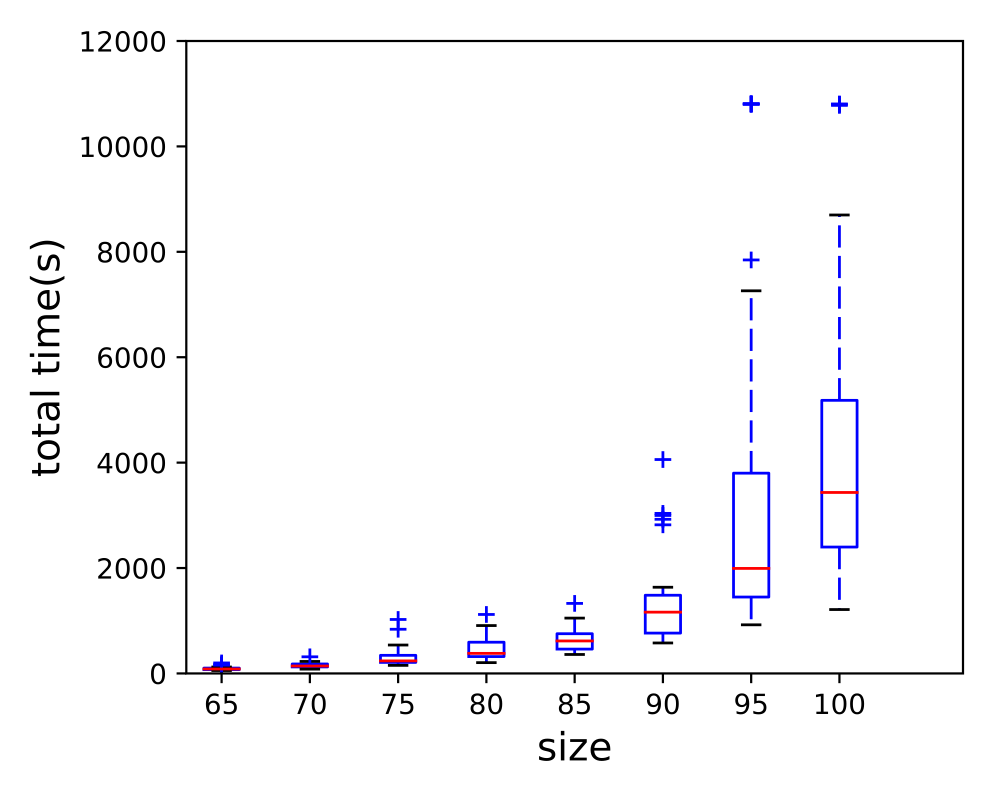}}}\\[-0.5ex]
	\subfloat[Barrier\label{timeT2K3}]{{\includegraphics[width=.45\textwidth]{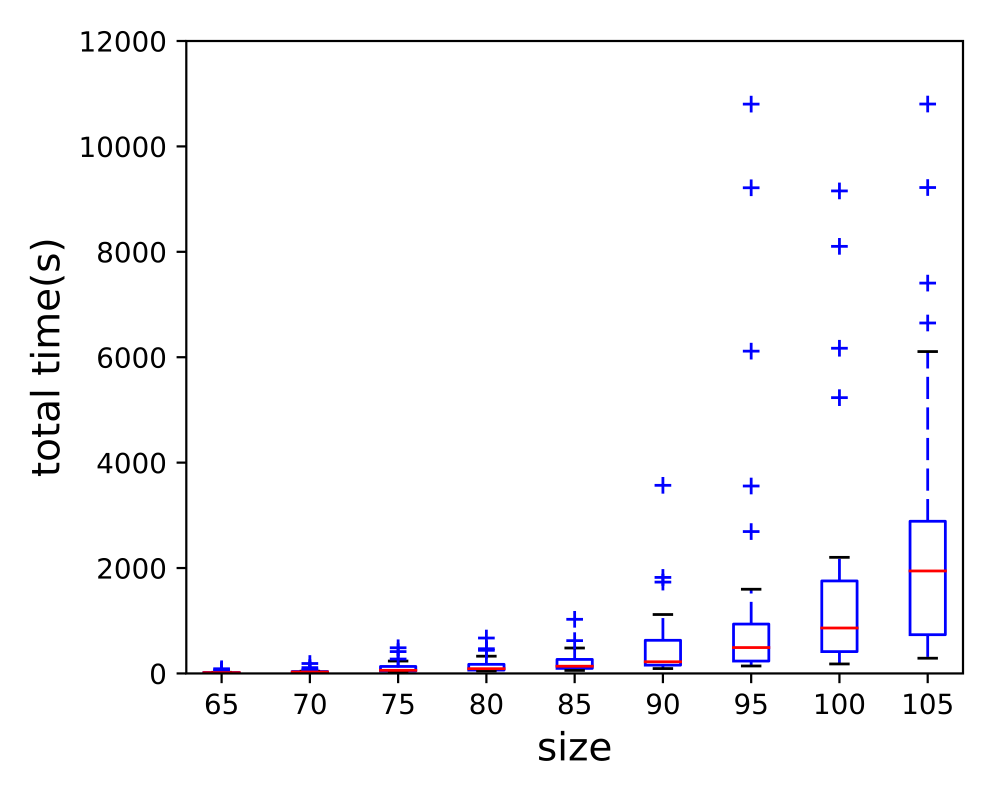}
	}}
	\caption{Total solving time for the \CGSP\ for different set stabilization methods.}
	\label{fig:timeStab}
\end{figure}

Finally, we compare different approaches to stabilize the column generation. As discussed in Section~\ref{sec:stab}, we studied three options: no stabilization, Weingetz Stabilization, and Barrier Methods. For the Weingetz Stabilization, using the {\tt irace} package\cite{irace}, we found that the best stabilization parameter is $\alpha=0.55$. The results are shown in Figure \ref{fig:timeStab}. As we can see, although no noticeable difference is observed when Weingetz method is used compared to not applying any stabilization at all, by replacing Simplex with the Barrier Method, we achieve a significant performance improvement. \asbC{Instances with 105 points were solved, on average$\;\pm\;$std-dev, in $2666 \pm 2711$ seconds when using the Barrier Method.}

As expected, when the entire model fits into memory, a better performance is obtained compared to running a column generation algorithm. This is due not only to the time spent solving the pricing problem many times per node, but also to the fact that modern ILP solvers are equipped with extremely powerful pre-processing routines that can reach their maximum potential when the models are completely loaded into memory.
However, the trade off between solution time and memory consumption, illustrated in Figure \ref{fig:memFinal}, reveal that column generation leads to significant savings in memory usage without drastic losses in performance. Thus, the technique is a good alternative when memory becomes the limiting factor to solving instances.

\begin{figure}[!htb]
	\centering
	\subfloat[\fullSP\label{memFull}]{{\includegraphics[width=.45\textwidth]{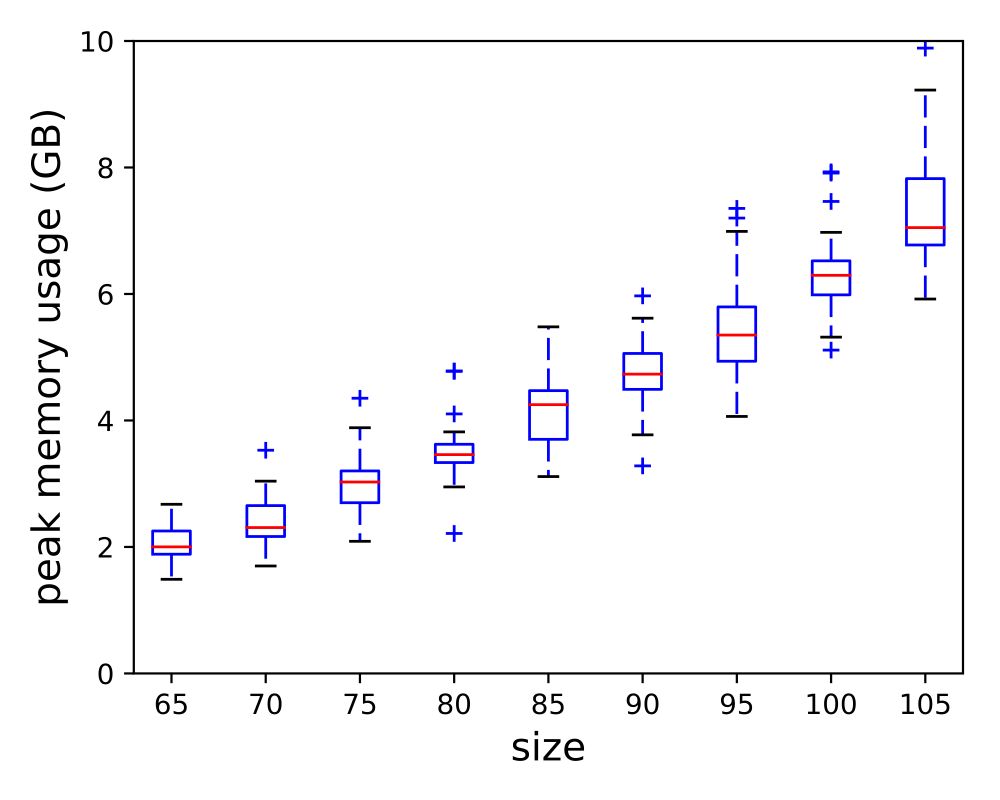}
	}}
	\hfill
	\subfloat[\CGSP Barrier\label{memBarrier}]{{\includegraphics[width=.45\textwidth]{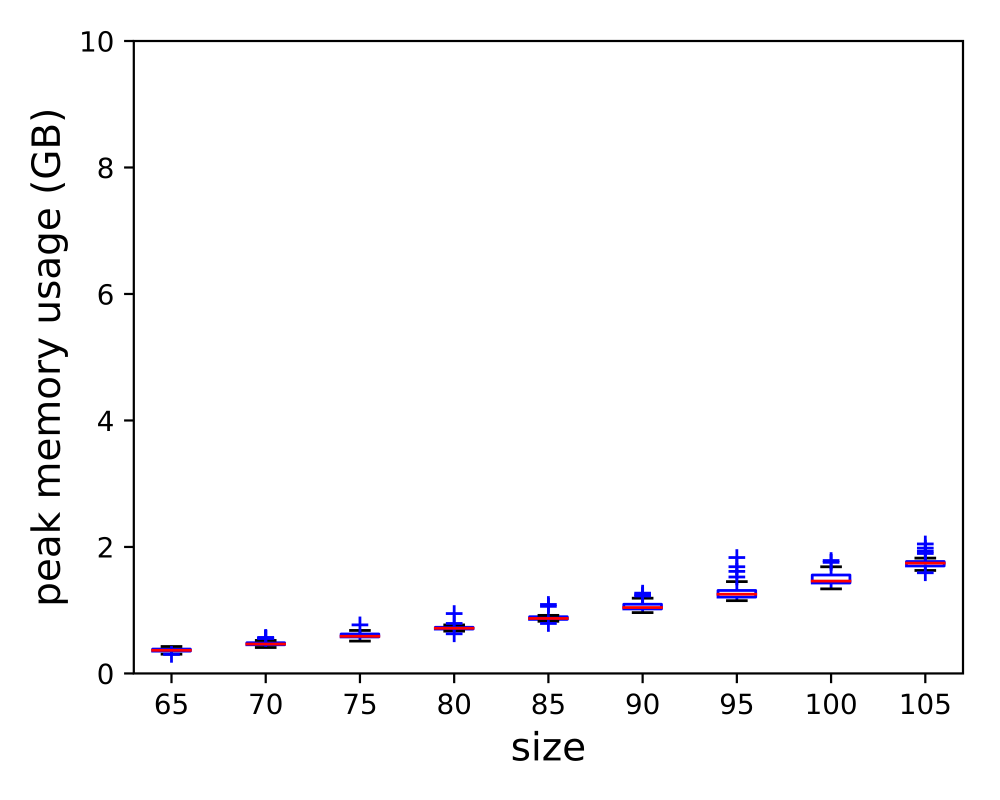}}}%
	\caption{Peak memory consumption for the two variations of the Set Partition Model: full and column generation using the barrier method.}
	\label{fig:memFinal}
\end{figure}

\section{Conclusions}
\label{sec:conclusions}

\cidC{\pjrC{In this paper, we discussed a new model for the \MCPP featuring variables assigned also to the convex polygons having vertices on the input point set, in contrast to the previous known formulation that only contained variables associated to edges with endpoints in that set. To cope with  the exponential number of variables, we proposed a column generation based algorithm and discussed implementation aspects that made it more efficient. The aspects investigated include the deployment of stabilization methods, which led to the use of the Barrier Method for solving linear relaxations, the development of a primal heuristic and of a simple yet effective branching rule. Also, a family of cuts inspired by a previously known ILP model was incorporated to the algorithm which, because of the dynamic and simultaneous inclusion of variables and constraints to the model, is characterized as a {\em branch-and-cut-and-price} algorithm. 
All those different aspects were assessed through a series of experiments to show their individual contribution to the algorithm's performance. As a result, constrained to identical computational resources, the new algorithm was able to solve instances with more than twice the size of what was possible in previous works.
Directions for future research include investigating a different stabilization model, and finding facet defining cuts.}}

\paragraph*{Acknowledgments}
{\small This work was supported in part by grants from: {\it Brazilian Nation\-al Council for Scientific and Technological Development} (CNPq),
~\#313329/2020-6, 
~\#309627/2017-6, 
~\#304727/2014-8; 
 {\it São Paulo Research Foundation} ({\sc Fapesp}), %
~\#2020 /09691-0, 
~\#2018/26434-0, 
~\#2018/14883-5. 
~\#2014/12236-1; 
 and {\it Fund for Support to Teaching, Research and Outreach Activities} (FAEPEX).}

{\small
\bibliography{main}
\bibliographystyle{abbrv}
}

\end{document}